\documentclass[11pt]{article}
\usepackage{amssymb}
\usepackage{colortbl}
\usepackage{amsfonts,amsmath, longtable}
\usepackage{ulem}
\usepackage{comment}

        \def\theequation{\thesection.\arabic{equation}}

\newcommand{\tr}{{\rm tr}}

\newcommand{\mL}{{\mathcal L}}
\newcommand{\mM}{{\mathcal M}}

\newcommand{\mH}{{\mathcal H}}

\newcommand{\la}{\lambda}
\newcommand{\vf}{\varphi}
\newcommand{\al}{\alpha}
\newcommand{\be}{\beta}
\newcommand{\ga}{\gamma}
\newcommand{\om}{\omega}
\newcommand{\vth}{\vartheta}

\newcommand{\bv}{{\bar{v}}}
\newcommand{\tv}{{\breve{v}}}
\newcommand{\tbv}{{\breve{\bar v}}}

\newcommand{\tnu}{{\breve{\nu}}}
\newcommand{\bnu}{{\bar{\nu}}}
\newcommand{\tbnu}{{\breve{\bar \nu}}}

\newcommand{\mC}{\mathbb C}
\newcommand{\mZ}{\mathbb Z}

\newtheorem{predl}{Proposition}[section]
\newtheorem{lemma}{Lemma}[section]

\newenvironment{proof}{\par\noindent{\bf Proof.}}{\hfill$\scriptstyle\blacksquare$}

\def\beq{\begin{equation}}
\def\eq{\end{equation}}
\def\p{\partial}

\newtheorem{theor}{Theorem}[section]

\newcommand{\mat}[4]{\left(\begin{array}{cc}{#1}&{#2}\\ \ \\{#3}&{#4}
\end{array}\right)}

\newcommand{\mats}[4]{\left(\begin{array}{cc}{#1}&{#2}\\ {#3}&{#4}
\end{array}\right)}

\def\res{\mathop{\hbox{Res}}\limits}

\begin{document}

\begin{center}

\setcounter{page}{1}

\vspace{0mm}

{\Large{\bf
Classical elliptic ${\rm BC}_1$ Ruijsenaars-van Diejen model:
}}

\vspace{3mm}

{\Large{\bf
relation to Zhukovsky-Volterra gyrostat and
}}

 \vspace{3mm}

{\Large{\bf 1-site classical XYZ model with boundaries}}

 \vspace{10mm}

 {\Large {A. Mostovskii}}
\qquad\qquad
 {\Large {A. Zotov}}

  \vspace{6mm}

 {\it Steklov Mathematical Institute of Russian
Academy of Sciences,\\ Gubkina str. 8, 119991, Moscow, Russia}




   \vspace{5mm}

 {\small\rm {E-mails: mostovskii.am21@physics.msu.ru, zotov@mi-ras.ru}}

\end{center}

\vspace{0mm}

\begin{abstract}
We present a description of the classical elliptic ${\rm BC}_1$ Ruijsenaars-van Diejen model
with 8 independent coupling constants
through a pair of ${\rm BC}_1$ type classical Sklyanin algebras generated
by the (classical) quadratic reflection equation with non-dynamical XYZ $r$-matrix. For this purpose, we consider the classical version
of the $L$-operator for the Ruijsenaars-van Diejen model proposed by O. Chalykh. In ${\rm BC}_1$ case
it is factorized to the product of two Lax matrices depending on 4 constants. Then we apply an IRF-Vertex type gauge transformation and obtain a product of the Lax matrices for
the Zhukovsky-Volterra gyrostats. Each of them is described by the ${\rm BC}_1$ version of the classical
Sklyanin algebra. In particular case, when 4 pairs of constants coincide, the ${\rm BC}_1$ Ruijsenaars-van Diejen model exactly coincides with the relativistic Zhukovsky-Volterra gyrostat.
Explicit change of variables is obtained.
We also consider another special case of the ${\rm BC}_1$
Ruijsenaars-van Diejen model with 7 independent constants. We show that it can be reproduced
by considering the transfer matrix of the classical 1-site XYZ chain with boundaries.
In the end of the paper, using another gauge transformation we represent the Chalykh's
Lax matrix in a form depending on the Sklyanin's generators.
\end{abstract}

%

{\small{ \tableofcontents }}

\bigskip


\section{Introduction}\label{sec1}
\setcounter{equation}{0}

We study the classical elliptic ${\rm BC}_1$ Ruijsenaars-van Diejen model \cite{vD}.
Its Hamiltonian can be represented in several different forms, see \cite{KH}.
Let us write down the form we use in this paper:
\beq\label{w01}
  \begin{array}{c}
  \displaystyle{
H^{\hbox{\tiny{8vD}}}=v(\eta,q)\bar v(\bar\eta,q)e^{p/c}+v(\eta,-q)\bar v(\bar\eta,-q)e^{-p/c}-2\sum_{a=0}^3\nu_a\bar\nu_a\wp(q+\om_a)\,,
}
 \end{array}
 \eq
where $p,q\in\mC$ are
canonically conjugated momenta and coordinate of particle, $\eta,\bar\eta,c,\nu_a,\bar\nu_a\in\mC$, $a=0,...,3$
are constants, $\wp(x)$ -- is the Weierstrass elliptic function on an elliptic curve $\mC/\mZ+\tau\mZ$
with elliptic moduli $\tau$ and $\om_a$ are four half-periods $0,1/2,1/2+\tau/2,\tau/2$.
The function $v$ is defined as
\beq\label{w02}
  \begin{array}{c}
  \displaystyle{
v(\eta,q|\nu)=v(\eta,q)=
\sum\limits_{a=0}^3\nu_a \exp(4\pi \imath \eta\p_\tau\om_a)\phi(2\eta,q+\om_a)\,,
\quad \phi(x,q)=\frac{\vth'(0)\vth(x+q)}{\vth(x)\vth(q)}
 }
 \end{array}
 \eq
with the first Jacobi theta function $\vth(x)$ (\ref{a01}), and ${\bar v}(\bar\eta,q)=v(\bar\eta,q|\bar\nu)$
is the same function with parameters $\bar\eta,\bar\nu_a$ instead of $\eta,\nu_a$.
The expression (\ref{w01}) depends on 10 coupling constants $\eta,\nu_a,\bar\eta,\bar\nu_a$
 (the eleventh constant $c$ is a relativistic parameter). Two of them can be fixed since the rescaling
  $\nu_a\to \kappa\nu_a$ yields $H^{\hbox{\tiny{8vD}}}\to \kappa H^{\hbox{\tiny{8vD}}}$
   and similarly for $\bar\nu_a$.
  Therefore, we have 8 independent constants (it is reflected in the notation $H^{\hbox{\tiny{8vD}}}$). It is convenient to keep all 10 constants
   for some particular and limiting cases.
The model (\ref{w01}) generalizes the elliptic Ruijsenaars-Schneider system \cite{RS} of ${\rm A}_1$ type (the 2-body case
in the center of mass frame), which has a single coupling constant.

The Hamiltonian (\ref{w01}) corresponds to the case of a 1-body ($n=1$) in the
elliptic $n$-body ${\rm BC}_n$ Ruijsenaars-van Diejen model with 9 coupling constants.
Its integrability at the quantum level was partially proved in \cite{vD} and then finally proved
by Y. Komori and K. Hikami in \cite{KH}
using Dunkl-Cherednik-type operators.
The Lax representation of size $2n\times 2n$ was constructed by O. Chalykh in \cite{Ch1} (see also \cite{Pusztai}
for degenerated trigonometric cases).
In the non-relativistic limit it reproduces the Lax matrix by K. Takasaki \cite{Ta} for the elliptic ${\rm BC}_n$
Calogero-Inozemtsev model with 5 constants \cite{Inoz89}. In the $n=1$ case the Hamiltonian of
${\rm BC}_1$
Calogero-Inozemtsev model is the non-relativistic limit of (\ref{w01})\footnote{The set of constants
$\nu_a$ in (\ref{w03}) does not coincide with $\nu_a$ from (\ref{w01}). The constants are redefined in
the non-relativistic limit, see the next Section.}:
\beq\label{w03}
  \begin{array}{c}
  \displaystyle{
H^{\hbox{\tiny{Inoz}}}=\frac{p^2}{2}-\sum\limits_{a=0}^3\nu_a^2\wp(q+\om_a)\,.
 }
 \end{array}
 \eq
In fact, this Hamiltonian is known for many years since it is the Hamiltonian for the elliptic form of the Painlev\'e VI equation
in time $\tau$ \cite{Painleve1906}, see also a review \cite{LOZ2014}. Let us also remark that
the (difference) q-Painlev\'e equations, in their turn, are related to the ${\rm BC}_1$ Ruijsenaars-van Diejen model under consideration \cite{Tak}.

It was shown in \cite{LOZ05} that equations of motion for the ${\rm BC}_1$
Calogero-Inozemtsev model (\ref{w03}) are equivalently written in the form of
Zhukovsky-Volterra gyrostat \cite{ZhV} (see also \cite{Basak}):
\beq\label{w04}
  \begin{array}{c}
  \displaystyle{
\dot{\vec{S}}=\vec{S}\times J(\vec{S})+\vec{S}\times \vec{\lambda}\,,
 }
 \end{array}
 \eq
where $\vec{S}=(S_1,S_2,S_3)$ is an angular momentum vector of rotation of rigid body in 3-dimensional space (in our consideration
the 3d space is complexified to $\mC^3$), $J(\vec{S})=(J_1S_1,J_2S_2,J_3S_3)$ with some constants
$J_1,J_2,J_3$ (these are the components of the inverse inertia tensor written in principle axes) and $\vec{\lambda}$ is a constant vector (related to the cyclical motion of the liquid in the cavity). Its components are certain linear
combinations of the constants $\nu_a$ (and the fourth independent linear combination
of  $\nu_a$ enters (\ref{w04}) through the length of $\vec{S}$). The above statement
on relation between (\ref{w03}) and (\ref{w04}) means existence of explicit change
of variables $S_k=S_k(p,q,\nu)$.
Similar relation was also observed and discussed in \cite{Tur} and \cite{Zhedanov}. The case $\vec{\lambda}=0$ corresponds to all pairwise equal $\nu_a$. Then
(\ref{w04}) turns into the Euler top, and the Calogero-Inozemtsev model (\ref{w03}) becomes the ${\rm A}_1$ elliptic
Calogero-Moser model \cite{Calogero2}. The relation between the Euler top and the $A$-type Calogero-Moser is known from
\cite{LOZ}, see also \cite{AtZ1}.

The next important result of \cite{LOZ05} is that the Zhukovsky-Volterra gyrostat (\ref{w04}) is a bihamiltonian
model. That is, besides the standard description through the linear Poisson brackets on ${\rm sl}_2^*$,
there is also a description through the quadratic Poisson brackets:
\beq\label{w05}
  \begin{array}{c}
  \displaystyle{
c\{S_\al,S_\be\}=-i\varepsilon_{\al\be\ga} S_0 S_\ga\,,\quad \al,\be,\ga\in\{1,2,3\}\,,
}
\\ \ \\
  \displaystyle{
c\{S_0,S_\al\}=-i\varepsilon_{\al\be\ga} S_\be S_\ga\Big(\wp(\om_\be)-\wp(\om_\ga)\Big)
+i\varepsilon_{\al\be\ga}(S_\be\lambda_\ga-\lambda_\be S_\ga)\,.
}
 \end{array}
 \eq
When $\vec\lambda=0$, it is the classical Sklyanin algebra \cite{Skl2}. In the generic case ($\vec\lambda\neq 0$)
we call it the (classical) ${\rm BC}_1$ Sklyanin algebra.
Also, we refer to this model as the relativistic Zhukovsky-Volterra gyrostat, where ''relativistic'' means that
the Poisson structure is quadratic\footnote{The non-relativistic models are described by linear $r$-matrix structures based on
Lie algebras, while relativistic models (e.g. the Ruijsenaars-Schneider model) are governed by quadratic $r$-matrix structure
similar to those for Lie groups.}.
The quantum version of the algebra (\ref{w05}) was described in \cite{LOZ05} as well using the
quantum reflection equation \cite{Skl-refl}. See also \cite{MS} for a study of this type generalizations of the Sklyanin algebra.

{\bf Purpose of the paper.} We consider $2\times 2$ Lax matrix $\mL^{\hbox{\tiny{Ch}}}(z)$ for the
${\rm BC}_1$ Ruijsenaars-van Diejen model (\ref{w01}) from \cite{Ch1}.
We mention that under simple transformation $\mL^{\hbox{\tiny{Ch}}}(z)$ takes the ''symmetric'' form $L(z){\bar L}(z)$,
which means that $L(z)=L(z|p,q,\eta,\nu_a)$, ${\bar L}(z)=L(z|p,q,\bar\eta,{\bar\nu}_a)$.
As a by-product, we deduce the (accompany) $M$-matrix for the classical Lax equation.
Our first result is
 related to a particular case of (\ref{w01}) when $\eta=\bar\eta$ and $\nu_a={\bar\nu}_a$, so that
the model (\ref{w01}) contains 4 independent constants in this case. Then
$L(z)={\bar L}(z)$ and the Hamiltonian (\ref{w01}) in this case (up to some constants) is equal to $(H_1^{\hbox{\tiny{4vD}}})^2$, where
\beq\label{w06}
  \begin{array}{c}
  \displaystyle{
H_1^{\hbox{\tiny{4vD}}}=v(\eta,q)e^{p/2c}+v(\eta,-q)e^{-p/2c}\,.
}
 \end{array}
 \eq
We prove that this model is gauge equivalent
to the relativistic Zhukovsky-Volterra gyrostat and find explicit change of
variables
\beq\label{w07}
  \begin{array}{c}
  \displaystyle{
S_0(p,q)=\frac1{2}\Big(v(\eta,q)e^{p/2c}+v(\eta,-q)e^{-p/2c}\Big)=\frac1{2}\,H_1^{\hbox{\tiny{4vD}}}\,,
}
 \end{array}
 \eq
 $$
  \displaystyle{
S_\al(p,q)={c_\al}\left[\frac{1}{2}\Big(v(\eta,q)e^{p/2c}-v(\eta,-q)e^{-p/2c}\Big)\vf_\al(2q)+
\tnu_0\vf_\be(2q)\vf_\ga(2q)+\vf_\al(2q)\sum\limits_{k=1}^3\tnu_k\vf_k(2q)\right]\,,
}
 $$
where $\vf_\al$ are the functions (\ref{a27}) and $\tnu_a$ are the linear combinations of $\nu_a$ (\ref{w2113}).
We prove that the map from $p,q$ variables to $S_a(p,q)$ is Poisson, that is the Poisson brackets $\{S_a(p,q),S_b(p,q)\}$
computed through the canonical brackets $\{p,q\}=1$ coincide with the classical ${\rm BC}_1$ Sklyanin algebra (\ref{w05}).
The gauge transformation generating (\ref{w07}) is given by the (intertwining) matrix $\Xi(z,q)$, which comes from
the IRF-Vertex correspondence in quantum models of statistical mechanics \cite{Baxter-irf}. See also \cite{AtZ1,LOZ,LOZ2014} for its application to classical models.


Next, we return to discussion of the model $H^{\hbox{\tiny{8vD}}}$ (\ref{w01}). Using the described above
results we easily get that the model (\ref{w01}) is gauge equivalent to a pair of the
relativistic Zhukovsky-Volterra gyrostats with sets of constants $\eta,\nu_a$ and $\bar\eta,\bnu_a$.
It is important that both gyrostats are defined on the same phase space since their dynamical variables
$S_a(p,q|\eta,\nu_a)$, ${\bar S}_b(p,q|\bar\eta,\bnu_a)$
are expressed through $p,q$ in the same way (\ref{w07}) but with different constants.
This is why we call this model as coupled gyrostats.
Unfortunately,
we do not know a nice answer for the mixed Poisson
brackets $\{S_a(p,q),{\bar S}_b(p,q)\}$ although they can be of course
calculated as functions of $p,q$.
Being lifted to the quantum level, this problem is
to find closed commutation relations for the Sklyanin type generators depending on
different parameters but acting on the same representation space. To our best knowledge,
solution is unknown. We hope to clarify this problem in our future works.

Then in Section \ref{sec5} we explain the relation of the ${\rm BC}_1$ Ruijsenaars-van Diejen model
to 1-site classical XYZ ''chain'' with boundaries defined by a pair of different constant
 $K$-matrices.
This time we deal with the original Sklyanin algebra without linear terms.
It is given by relations (\ref{w05}) with $\lambda_\al=0$.
We show that the classical transfer matrix leads to the Hamiltonian that
can be represented in the form
of (\ref{w01}) with $\eta=\bar\eta$.

Finally, in the last Section \ref{sec6} we perform another gauge transformation with
the Chalykh's Lax matrix and obtain an expression depending on the original Sklyanin's generators.
Necessary elliptic functions properties are given in the Appendix.

\section{Elliptic ${\rm BC}_1$  Ruijsenaars-van Diejen model}\label{sec2}
\setcounter{equation}{0}

\subsection{Notations}
Below we consider the model (\ref{w01}) with the canonically conjugated momentum $p$ and position of particle $q$:
\beq\label{w200}
  \begin{array}{c}
  \displaystyle{
\{p,q\}=1\,.
 }
 \end{array}
 \eq
In (\ref{w01}) and in what follows we use the function
\beq\label{w210}
  \begin{array}{c}
  \displaystyle{
v(z,u|\nu)=v(z,u)=\sum\limits_{a=0}^3\nu_a\vf_a(2z,u+\om_a)=
\sum\limits_{a=0}^3\nu_a \exp(4\pi \imath z\p_\tau\om_a)\phi(2z,u+\om_a)\,,
 }
 \end{array}
 \eq
which has the property
\beq\label{w211}
  \begin{array}{c}
  \displaystyle{
v(z,u|\nu)=v(u,z|\breve\nu)=\sum\limits_{a=0}^3\breve\nu_a\vf_a(2u,z+\om_a)\,,
 }
 \end{array}
 \eq
where
  \beq\label{w2113}
 \begin{array}{c}
  \left(\begin{array}{c}
 \tnu_0
 \\
 \tnu_1
 \\
 \tnu_2
 \\
 \tnu_3
 \end{array}\right)
 =
   \displaystyle{\frac12}
 \left(\begin{array}{cccc}
 1 & 1 & 1 & 1
 \\
  1 & 1 & -1 & -1
 \\
  1 & -1 & 1 & -1
 \\
  1 & -1 & -1 & 1
 \end{array}\right)
  \left(\begin{array}{c}
 \nu_0
 \\
 \nu_1
 \\
 \nu_2
 \\
 \nu_3
 \end{array}\right)\,.
 \end{array}
 \eq
 The property (\ref{w211}) follows from (\ref{a23}). The inverse transformation (from $\tnu_k$ to $\nu_k$)
 is the same as the direct one due to (\ref{a25}).
For shortness we use notation
\beq\label{w212}
  \begin{array}{c}
  \displaystyle{
v(u,z|\breve\nu)=\breve v(u,z)\,.
 }
 \end{array}
 \eq
That is
\beq\label{w213}
  \begin{array}{c}
  \displaystyle{
v(z,u)=\tv(u,z)\,.
 }
 \end{array}
 \eq
One more short notation is used for the function $v$ depending on the set of constants $\bnu_a$:
\beq\label{w214}
  \begin{array}{c}
  \displaystyle{
\bv(z,u)=v(z,u|\bnu)=\sum\limits_{a=0}^3\bnu_a\vf_a(2z,u+\om_a)\,.
 }
 \end{array}
 \eq

\subsection{Lax pair}
In \cite{Ch1} O. Chalykh proposed the following Lax matrix:
\beq\label{w230}
  \begin{array}{c}
  \displaystyle{
\mL^{\hbox{\tiny{Ch}}}(z)=\mat{v(\eta,q)}{-v(z,q)}{-v(z,-q)}{v(\eta,-q)}
\mat{\bv(\bar\eta,q)e^{p/c}}{-\bv(z,q)}{-\bv(z,-q)}{\bv(\bar\eta,-q)e^{-p/c}}\,.
 }
 \end{array}
 \eq
Let us explain how the Hamiltonian $H^{\hbox{\tiny{8vD}}}$ (\ref{w01}) follows from (\ref{w230}).
Consider $\tr \mL^{\hbox{\tiny{Ch}}}(z)$
\beq\label{w2301}
  \begin{array}{c}
  \displaystyle{
\tr \mL^{\hbox{\tiny{Ch}}}(z)=v(\eta,q)\bar v(\bar\eta,q)e^{p/c}+v(\eta,-q)\bar v(\bar\eta,-q)e^{-p/c}+v(z,q)\bar v(z,-q)+v(z,-q)\bar v(z,q)\,.
 }
 \end{array}
 \eq
Two last terms contain dependence on the spectral parameter $z$. In order to simplify this expression
we use identity (\ref{vbarv2}).
Plugging (\ref{vbarv2}) into (\ref{w2301}) we get
\beq\label{w2302}
  \begin{array}{c}
  \displaystyle{
\tr \mL^{\hbox{\tiny{Ch}}}(z)=H^{\hbox{\tiny{8vD}}}+
2\sum_{\alpha=0}^3\breve\nu_\alpha\breve{\bar\nu}_\alpha\wp(z+\om_\al)\,,
 }
 \end{array}
 \eq
where the last sum is independent of dynamical variables.

\paragraph{Symmetric form for $\mL^{\hbox{\tiny{Ch}}}(z)$.} Let us slightly modify the Lax matrix (\ref{w230}).
Notice that
\beq\label{w231}
  \begin{array}{c}
  \displaystyle{
\mat{\bv(\bar\eta,q)e^{p/c}}{-\bv(z,q)}{-\bv(z,-q)}{\bv(\bar\eta,-q)e^{-p/c}}=
}
\\ \ \\
  \displaystyle{
=\mat{e^{p/2c}}{0}{0}{-e^{-p/2c}}
\mat{\bv(\bar\eta,q)}{\bv(z,q)}{\bv(z,-q)}{\bv(\bar\eta,-q)}
\mat{e^{p/2c}}{0}{0}{-e^{-p/2c}}\,.
 }
 \end{array}
 \eq
Perform the gauge transformation
\beq\label{w232}
  \begin{array}{c}
  \displaystyle{
\mL^{\hbox{\tiny{Ch}}}(z)\to \mL(z)=\mat{e^{p/4c}}{0}{0}{-e^{-p/4c}}\mL^{\hbox{\tiny{Ch}}}(z)\mat{e^{-p/4c}}{0}{0}{-e^{p/4c}}\,.
}
 \end{array}
 \eq
Using (\ref{w231}) it is easy to see that the result of transformation (\ref{w232}) is represented in the form:
\beq\label{w233}
  \begin{array}{c}
\framebox{$\mL(z)=L(z){\bar L}(z)$}\,,
 \end{array}
 \eq
where
\beq\label{w234}
  \begin{array}{c}
  \displaystyle{
{L}(z)=\mat{v(\eta,q)e^{p/2c}}{v(z,q)}{v(z,-q)}{v(\eta,-q)e^{-p/2c}}
}
 \end{array}
 \eq
and
\beq\label{w235}
  \begin{array}{c}
  \displaystyle{
{\bar L}(z)=\mat{\bv(\bar\eta,q)e^{p/2c}}{\bv(z,q)}{\bv(z,-q)}{\bv(\bar\eta,-q)e^{-p/2c}}\,.
}
 \end{array}
 \eq
It is interesting to notice that the expression $\tr\mL(z)=\tr(L(z){\bar L}(z))$ provides the same Hamiltonian
as $\tr L(z)\tr{\bar L}(z)$ when $\bar\eta=\eta$. Introduce
\beq\label{w2351}
  \begin{array}{c}
  \displaystyle{
H_1^{\hbox{\tiny{4vD}}} =\tr{L}(z)=v(\eta,q)e^{p/2c}+v(\eta,-q)e^{-p/2c}
}
 \end{array}
 \eq
and
\beq\label{w2352}
  \begin{array}{c}
  \displaystyle{
{\bar H}_1^{\hbox{\tiny{4vD}}} =\tr{\bar L}(z)=
\bv(\bar\eta,q)e^{p/2c}+\bv(\bar\eta,-q)e^{-p/2c}\,.
}
 \end{array}
 \eq
\begin{predl}
The following relation holds true for ${H}^{\hbox{\tiny{8vD}}}$ (\ref{w01}) and
${H}_1^{\hbox{\tiny{4vD}}}$, ${\bar H}_1^{\hbox{\tiny{4vD}}}$ (\ref{w2351})-(\ref{w2352}):
\beq\label{w2353}
  \begin{array}{c}
\framebox{${H}^{\hbox{\tiny{7vD}}}={H}^{\hbox{\tiny{8vD}}}\Big|_{\bar\eta=\eta}=
{H}_1^{\hbox{\tiny{4vD}}}{\bar H}_1^{\hbox{\tiny{4vD}}}\Big|_{\bar\eta=\eta}+{\rm const}$}
 \end{array}
 \eq
\end{predl}
\begin{proof}
From (\ref{w2351})-(\ref{w2352}) we have:
\beq\label{w2354}
  \begin{array}{c}
  \displaystyle{
  {H}_1^{\hbox{\tiny{4vD}}}{\bar H}_1^{\hbox{\tiny{4vD}}}=
  v(\eta,q)\bar v(\eta,q)e^{p/c}+v(\eta,-q)\bar v(\eta,-q)e^{-p/c}
  +v(\eta,q)\bv(\eta,-q)+v(\eta,-q)\bv(\eta,q)\,.
}
 \end{array}
 \eq
In order to get the desired result one should use the identity (\ref{vbarv2}) for the last
two terms in (\ref{w2354}).
\end{proof}

\paragraph{$M$-matrix.} The Hamiltonian (\ref{w01}) provides the following equations of motion:
\beq\label{w243}
 \begin{array}{c}
  \displaystyle{
\dot q = \{H^{\hbox{\tiny{8vD}}},q\}= \partial_p H^{\hbox{\tiny{8vD}}}=\frac{1}{c}(v(\eta,q)\bar v(\bar\eta,q)e^{p/c}-v(\eta,-q)\bar v(\bar\eta,-q)e^{-p/c})
}
\end{array}
\eq
and
\beq\label{w2431}
 \begin{array}{c}
  \displaystyle{
\dot p = \{H^{\hbox{\tiny{8vD}}},p\}= -\partial_q H^{\hbox{\tiny{8vD}}}=
}
\\
  \displaystyle{
= -(v(\eta,q)\bar v(\bar\eta,q))'e^{p/c}+(v(\eta,-q)\bar v(\bar\eta,-q))'e^{-p/c}+2\sum_{\alpha=0}^3\nu_\alpha\bar\nu_\alpha\wp'(q+\om_\alpha)\,.
  }
\end{array}
\eq
In \cite{Ch1} the (accompany) $M$-matrices were defined in a general implicit form.
Below we give an explicit formula for $M$-matrix related to the Lax matrix (\ref{w233}).
Similar result for (\ref{w230}) was obtained in \cite{Mos}.
\begin{predl}
Equations (\ref{w243}), (\ref{w2431}) are equivalently written in the Lax form
\beq\label{w2440}
  \begin{array}{c}
  \displaystyle{
\dot\mL(z)=[\mL(z),\mM(z)]\,,
}
 \end{array}
 \eq
where $\mM(z)$ is the following matrix:
\beq\label{w244}
  \begin{array}{c}
  \displaystyle{
{\mM}(z)=\frac{1}{c}\mat{\mM_{11}(z,q,p)}{\mM_{12}(z,q,p)}{\mM_{21}(z,q,p)}{\mM_{22}(z,q,p)}
}
 \end{array}
 \eq
 with the entries
 \beq\label{w245}
 \begin{array}{l}
  \displaystyle{
\mM_{11}(z,q,p)=-v'(\eta,q)\bar v(\bar\eta,q)e^{p/c}+\bar v'(z,-q) v(z,q)+\frac14\,\partial_qH^{\hbox{\tiny{8vD}}}\,,
  }
  \\ \ \\
   \displaystyle{
\mM_{12}(z,q,p)=v(\eta,q)\bar v'(z,q)e^{p/2c}+\bar v(\bar\eta,-q)v'(z,q)e^{-p/2c}\,,
  }
  \\ \ \\
   \displaystyle{
\mM_{21}(z,q,p)= v(\eta,-q)\bar v'(z,-q)e^{-p/2c}+\bar v(\bar\eta,q)v'(z,-q)e^{p/2c}\,,
  }
  \\ \ \\
   \displaystyle{
\mM_{22}(z,q,p)=-v'(\eta,-q)\bar v(\bar\eta,-q)e^{-p/c}+\bar v'(z,q) v(z,-q)-\frac14\,\partial_qH^{\hbox{\tiny{8vD}}}\,.
  }
  \end{array}
 \eq
\end{predl}
\begin{proof}
Notice that
 \beq\label{w246}
\begin{array}{c}
\displaystyle{
    \mL_{22}(z,q,p) = \mL_{11}(z,-q,-p), \quad \mL_{21}(z,q,p) = \mL_{12}(z,-q,-p),
    }
    \\ \ \\
\displaystyle{
    \mM_{22}(z,q,p) = \mM_{11}(z,-q,-p), \quad \mM_{21}(z,q,p)=\mM_{12}(z,-q,-p).
    }
\end{array}
 \eq
 Therefore, we need to verify the matrix elements "11" and "12" of the Lax equation only.
 Also, it follows from the Lax equation that ${\dot H}^{\hbox{\tiny{8vD}}}=0$, that is $\dot q = -\dot p\,\frac{\p_pH^{\hbox{\tiny{8vD}}}}{\p_qH^{\hbox{\tiny{8vD}}}}$.

 For the matrix element "11" in the l.h.s. of the Lax equation (\ref{w2440}) we have
\beq\label{w247}
\begin{array}{c}
\displaystyle{
 \dot\mL_{11}(z)=\frac{\dot p}{c} v(\eta,q)\bar v(\bar\eta,q)e^{p/c}+\dot q \Big(\p_q(v(\eta,q)\bar v(\bar\eta,q))e^{p/c}+\p_q(v(z,q)\bar v(z,-q))\Big)=
 }
 \\ \ \\
 \displaystyle{
 =\dot p \Big[\frac{1}{c} v(\eta,q)\bar v(\bar\eta,q)e^{p/c}-\frac{\p_pH^{\hbox{\tiny{8vD}}}}{\p_qH^{\hbox{\tiny{8vD}}}}\Big(\p_q(v(\eta,q)\bar v(\bar\eta,q))e^{p/c}+\p_q(v(z,q)\bar v(z,-q))\Big)\Big]=
 }
 \\ \ \\
 \displaystyle{
=\frac{1}{c}\frac{\dot p}{\p_qH^{\hbox{\tiny{8vD}}}}\Big[e^{p/c}v(\eta,q)\bar v(\bar\eta,q)\p_q(v(z,-q)\bar v(z,q))+e^{-p/c}v(\eta,-q)\bar v(\bar\eta,-q)\p_q(v(z,q)\bar v(z,-q))+}
\\ \ \\
\displaystyle{
+\p_q(v(\eta,q)\bar v(\bar\eta,q)v(\eta,-q)\bar v(\bar\eta,-q))\Big]\,.
}
\end{array}
\eq
In order to obtain (\ref{w2431}) we need to check that the expression in the last square brackets is equal to

 \beq\label{w248}
\begin{array}{c}
\displaystyle{
-\mL_{12}(z,q,p)\mM_{21}(z,q,p)+\mL_{21}(z,q,p)\mM_{12}(z,q,p) =
}
\end{array}
\eq
$$
\begin{array}{c}
\displaystyle{
-\Big(v(\eta,q)\bar v(z,q)e^{p/2c}+\bar v(\bar\eta,-q)v(z,q)e^{-p/2c}\Big)\Big(v(\eta,-q)\bar v'(z,-q)e^{-p/2c}+\bar v(\bar\eta,q)v'(z,-q)e^{p/2c}\Big)+
}
\\ \ \\
\displaystyle{
+\Big(v(\eta,-q)\bar v(z,-q)e^{-p/2c}+\bar v(\bar\eta,q)v(z,-q)e^{p/2c}\Big)\Big(v(\eta,q)\bar v'(z,q)e^{p/2c}+\bar v(\bar\eta,-q)v'(z,q)e^{-p/2c}\Big)=
}
\\ \ \\
\displaystyle{
=e^{p/c}v(\eta,q)\bar v(\bar\eta,q)\p_q(v(z,-q)\bar v(z,q))+e^{-p/c}v(\eta,-q)\bar v(\bar\eta,-q)\p_q(v(z,q)\bar v(z,-q))+
}
\\ \ \\
\displaystyle{
+\bar v(\bar\eta,q)\bar v(\bar\eta,-q)\p_q(v(z,q)v(z,-q))+v(\eta,q) v(\eta,-q)\p_q(\bar v(z,q)\bar v(z,-q))=
}
\\ \ \\
\displaystyle{
=e^{p/c}v(\eta,q)\bar v(\bar\eta,q)\p_q(v(z,-q)\bar v(z,q))+e^{-p/c}v(\eta,-q)\bar v(\bar\eta,-q)\p_q(v(z,q)\bar v(z,-q))+
}
\\ \ \\
\displaystyle{
+\bar v(\bar\eta,q)\bar v(\bar\eta,-q)\p_q(v(\eta,q)v(\eta,-q))+v(\eta,q) v(\eta,-q)\p_q(\bar v(\bar\eta,q)\bar v(\bar\eta,-q))=
}
\\ \ \\
\displaystyle{
=e^{p/c}v(\eta,q)\bar v(\bar\eta,q)\p_q(v(z,-q)\bar v(z,q))+e^{-p/c}v(\eta,-q)\bar v(\bar\eta,-q)\p_q(v(z,q)\bar v(z,-q))+
}
\\ \ \\
\displaystyle{
+\p_q({\bar v}(\bar\eta,q)\bar v(\bar\eta,-q)v(\eta,q)v(\eta,-q))\,.
}
\end{array}
$$
Here we used that $\p_q(v(z,q)v(z,-q))$ is independent of $z$ due to the property (\ref{a50}).
By comparing (\ref{w248}) and (\ref{w247}),
 we obtain $\dot p = -\p_q H^{\hbox{\tiny{8vD}}}$ and, therefore, $\dot q = \p_p H^{\hbox{\tiny{8vD}}}$.

It remains to verify
\beq\label{w249}
\dot\mL_{12}(z)-\frac{1}{c}\Big(\mL_{11}(z)\mM_{12}(z)+\mL_{12}(z)\mM_{22}(z)-
\mM_{11}(z)\mL_{12}(z)-\mM_{12}(z)\mL_{22}(z)\Big)=0\,.
\eq
This calculation is straightforward and uses the same properties for the function $v(z,q)$.
\end{proof}

\subsection{4-constants ${\rm BC}_1$ Ruijsenaars-van Diejen model}

Consider the special case when $\nu_{\al}=\bar\nu_{\al}$ and $\eta=\bar\eta$. Then the $BC_1$ Ruijsenaars-van-Diejen model is described by the Lax matrix (\ref{w234}).
%
The Hamiltonian in this case takes the form
\beq\label{w253}
 \begin{array}{c}
  \displaystyle{
H^{\hbox{\tiny{4vD}}}_1 = \tr L(z)
=v(\eta,q)e^{p/2c}+v(\eta,-q)e^{-p/2c}\,.
  }
\end{array}
\eq
The equations of motion are written in the standard way:
\beq\label{w254}
 \begin{array}{c}
  \displaystyle{
\dot q = \{H^{\hbox{\tiny{4vD}}}_1,q\}= \partial_p H^{\hbox{\tiny{4vD}}}_1=\frac{1}{2c}(v(\eta,q)e^{p/2c}-v(\eta,-q)e^{-p/2c})
}
\end{array}
\eq
\beq\label{w2541}
 \begin{array}{c}
  \displaystyle{
\dot p = \{H^{\hbox{\tiny{4vD}}}_1,p\}= -\partial_q H^{\hbox{\tiny{4vD}}}_1= -v'(\eta,q)e^{p/2c}+v'(\eta,-q)e^{-p/2c}
  }
\end{array}
\eq
Similarly, for the second flow we have:
\beq\label{w2531}
 \begin{array}{c}
  \displaystyle{
H^{\hbox{\tiny{4vD}}}_2
=\frac12\,v(\eta,q)^2e^{p/c}+\frac12\,v(\eta,-q)^2e^{-p/c}-\sum_{\alpha=0}^3\nu_\alpha^2\wp(q+\om_\alpha)\,,
  }
\end{array}
\eq
which is calculated from $\frac12\,\tr L^2(z)$. Then
\beq\label{w25410}
 \begin{array}{c}
  \displaystyle{
\dot q = \{H^{\hbox{\tiny{4vD}}}_2,q\}= \partial_p H^{\hbox{\tiny{4vD}}}_2=\frac{1}{2c}(v(\eta,q)^2e^{p/c}-v(\eta,-q)^2e^{-p/c})\,,
}
\end{array}
\eq
\beq\label{w25411}
 \begin{array}{c}
  \displaystyle{
\dot p = \{H^{\hbox{\tiny{4vD}}}_2,p\}= -\partial_q H^{\hbox{\tiny{4vD}}}_2=
 }
\\
  \displaystyle{
=-v(\eta,q)v'(\eta,q)e^{p/c}+v(\eta,-q)v'(\eta,-q)e^{-p/c}+\sum_{\alpha=0}^3\nu_\alpha^2\wp'(q+\om_\alpha)\,.
  }
\end{array}
\eq
\begin{predl}
    M-matrices for the flows  (\ref{w253})-(\ref{w2541}) and (\ref{w2531})-(\ref{w25411}) are
as follows:
\beq\label{w255}
  \begin{array}{c}
  \displaystyle{
{M_1}(z)=\frac{1}{2c}\mat{0}{v'(z,q)}{v'(z,-q)}{0}\,,
}
 \end{array}
 \eq
 \beq\label{w2551}
  \begin{array}{c}
  \displaystyle{
{M_2}(z)=\frac{1}{2c}\Big(v(\eta,q)e^{p/2c}+v(\eta,-q)e^{-p/2c}\Big)\mat{0}{v'(z,q)}{v'(z,-q)}{0}\,.
}
 \end{array}
 \eq
\end{predl}
The proof is direct. It is based on the property (\ref{a54}).

\subsection{${\rm A}_1$ Ruijsenaars-Schneider model and non-relativistic limit}

\paragraph{Reproducing ${\rm A}_1$ Ruijsenaars-Schneider model.}
%
%
Consider the Lax matrix for the two-particle Ruijsenaars-Schneider model. It can be written as follows:
%
%
%
\beq\label{w252}
  \begin{array}{c}
  \displaystyle{
{{L}^{\hbox{\tiny{RS}}}}(z)=\mat{\phi(\eta,q_{12})e^{p_1/c}}{\phi(z,q_{12})e^{p_2/c}}
{\phi(z,q_{21})e^{p_1/c}}{\phi(\eta,q_{21})e^{p_2/c}}\,.
}
 \end{array}
 \eq
 It is easy to verify that the total momentum $p_1+p_2$ is a conservation law for the system with the Lax matrix (\ref{w252}). Therefore, we can switch to the center of mass frame: $p_1+p_2=0$, $q_1+q_2=0$.
 In this frame, after variables renaming $q=q_1$, $p=2p_1$ and additional conjugation by
 ${\rm diag}\{e^{p/(4c)}, e^{-p/(4c)}\}$, the matrix (\ref{w252}) is rewritten as
 \beq\label{w2521}
  \begin{array}{c}
  \displaystyle{
{{L}^{\hbox{\tiny{RS}}}}(z)\to
\mat{\phi(\eta,2q)e^{p/(2c)}}{\phi(z,2q)}{\phi(z,-2q)}{\phi(\eta,-2q)e^{-p/(2c)}}\,.
}
 \end{array}
 \eq
 Finally, we notice that (\ref{w2521}) coincides with (\ref{w234}) in the case $\breve\nu_1=\breve\nu_2=\breve\nu_3=0$, $\breve\nu_0=1$.

\paragraph{Non-relativistic limit to ${\rm BC}_1$ Calogero-Inozemtsev model.}
Consider the Lax matrix $L(z)$ (\ref{w234}). Similarly to the non-relativistic limit in the Ruijsenaars-Schneider model one should consider the limit $\eta\to 0$ and $c\to\infty$ simultaneously. Let $\eta=\kappa/c$ with some finite $\kappa\in\mC$.
For the matrix element $L_{11}(z)$ in the limit $c\to\infty$ using the local expansion (\ref{a08}) we have
\beq\label{w236}
  \begin{array}{c}
  \displaystyle{
v(\eta,q)e^{p/2c}=\tv(q,\eta)e^{p/2c}=\Big(\frac{\tnu_0}{\eta}+\tnu_0E_1(2q)
+\sum\limits_{k=1}^3\tnu_k\vf_k(2q)+O(\eta)\Big)
\Big(1+\frac{p}{2c}+O(1/c^2)\Big)=
}
\\ \ \\
  \displaystyle{
=\frac{c\tnu_0}{\kappa}+\frac{\tnu_0}{\kappa}\Big(\frac{p}{2}+{\kappa}E_1(2q)+
\frac{\kappa}{\tnu_0}\sum\limits_{k=1}^3\tnu_k\vf_k(2q)\Big)+O(\frac{1}{c})\,.
  }
 \end{array}
 \eq
By performing redefinition
\beq\label{w237}
  \begin{array}{c}
  \displaystyle{
p\to p-{2\kappa}E_1(2q)-
\frac{2\kappa}{\tnu_0}\sum\limits_{k=1}^3\tnu_k\vf_k(2q)
}
 \end{array}
 \eq
and
\beq\label{w238}
  \begin{array}{c}
  \displaystyle{
\tnu_a\to \frac{\kappa}{\tnu_0}\,\tnu_a\,,\quad a=0,1,2,3\,,
}
 \end{array}
 \eq
we get
\beq\label{w239}
  \begin{array}{c}
  \displaystyle{
L(z)=c1_{2\times 2}+L^{\hbox{\tiny{Inoz}}}(z)+O(\frac{1}{c})\,,
}
 \end{array}
 \eq
where
\beq\label{w240}
  \begin{array}{c}
  \displaystyle{
L^{\hbox{\tiny{Inoz}}}(z)=\mat{p/2}{\tv(q,z)}{\tv(-q,z)}{-p/2}
}
 \end{array}
 \eq
and $\tnu_0$ is replaced with $\kappa$. We keep the notation $\tnu_0$ by denoting back $\kappa:=\tnu_0$
in the final answer.

Obviously, for the Lax matrix $\mL(z)=L(z){\bar L}(z)$ (\ref{w233}) in the limit we again obtain
 the Lax matrix of the form $L^{\hbox{\tiny{Inoz}}}(z)$
\beq\label{w241}
  \begin{array}{c}
  \displaystyle{
\mL(z)=c^21_{2\times 2}+c\cdot{\rm const}\cdot L^{\hbox{\tiny{Inoz}}}(z)+O(1)
}
 \end{array}
 \eq
with an appropriate redefinitions of constants (the constants $\tnu_a$ and $\tbnu_a$ are combined into
sums $\tnu_a+\tbnu_a$).


\section{Zhukovsky-Volterra gyrostat}\label{sec3}
\setcounter{equation}{0}

In this Section we recall description of the Zhukovsky-Volterra gyrostat following \cite{LOZ05}.

\subsection{Linear Poisson structure}

The Lax matrix for the Zhukovsky-Volterra gyrostat has the form:
\beq\label{w310}
  \begin{array}{c}
  \displaystyle{
L^{\hbox{\tiny{ZhV}}}(z)=
\sum\limits_{\al=1}^3 \Big(S_\al\vf_\al(z)-\frac{\lambda_\al}{\vf_\al(z)}\Big)\sigma_{4-\al}\,,
}
 \end{array}
 \eq
where
$\sigma_k$ are the Pauli matrices\footnote{Notice that the numeration of dynamical variables $S_\al$ and the functions $\vf_\al(z)$ is non-standard in (\ref{w310}).
The numbers of $\sigma_1$ and $\sigma_3$ components are interchanged. This comes from the previous
consideration of the van Diejen model, where the numeration (\ref{a222}) (and (\ref{a27})) is used.}
 \beq\label{w311}
 \begin{array}{c}
  \displaystyle{
  \sigma_1=\mats{0}{1}{1}{0}\,,\quad
  \sigma_2=\mats{0}{-\imath}{\imath}{0}\,,\quad
   \sigma_3=\mats{1}{0}{0}{-1}
  }
 \end{array}
\eq
and below we also use notation
 $\sigma_0=\sigma_4$ for the identity $2\times 2$ matrix.
The Poisson structure is given by the Poisson-Lie brackets on ${\rm sl}_2^*$:
 \beq\label{w312}
 \begin{array}{c}
  \displaystyle{
 \{S_\al,S_\be\}=-\imath\varepsilon_{\al\be\ga}S_\ga\,,\quad \al,\be,\ga\in\{1,2,3\}\,.
  }
 \end{array}
\eq
Using (\ref{a30}) one gets
 \beq\label{w3121}
 \begin{array}{c}
  \displaystyle{
 \frac14\,\tr\Big(\Big(L^{\hbox{\tiny{ZhV}}}(z)\Big)^2\Big)=C\wp(z)+H^{\hbox{\tiny{ZhV}}}
 +\frac12\sum\limits_{\al=1}^3\frac{\lambda_\al^2}{\vf_\al(z)^2}\,,
  }
 \end{array}
\eq
 where
 \beq\label{w3122}
 \begin{array}{c}
  \displaystyle{
 C=\frac12\Big(S_1^2+S_2^2+S_3^2\Big)
  }
 \end{array}
\eq
 is the Casimir function of the brackets (\ref{w312}) and
 $H^{\hbox{\tiny{ZhV}}}$ is the Hamiltonian:
 \beq\label{w3123}
 \begin{array}{c}
  \displaystyle{
H^{\hbox{\tiny{ZhV}}}=-\frac12\sum\limits_{\al=1}^3 S_\al^2\wp(\om_\al)-\sum\limits_{\al=1}^3S_\al\lambda_\al\,.
  }
 \end{array}
\eq
 Introduce the set of matrices:
 \beq\label{w3124}
 \begin{array}{c}
  \displaystyle{
S=\sum\limits_{\al=1}^3 S_\al\sigma_{4-\al}\,,\qquad
\hat\wp(S)=-\frac{1}{2}\sum\limits_{\al=1}^3 S_\al\wp(\om_\al)\sigma_{4-\al}\,,\qquad
\lambda=\sum\limits_{\al=1}^3 \lambda_\al\sigma_{4-\al}\,.
  }
 \end{array}
\eq
 Then the equations of motion generated by the Hamiltonian (\ref{w3123}) and the Poisson brackets (\ref{w312})
 take the form\footnote{Equations (\ref{w3125}) differ from (\ref{w04}) by the factor $2\imath$ entering commutator
 of Pauli matrices. To avoid it one may rescale the Hamiltonian.}
 \beq\label{w3125}
 \begin{array}{c}
  \displaystyle{
{\dot S}=[S,\hat\wp(S)]-\frac12\,[S,\lambda]\,.
  }
 \end{array}
\eq
 These equations are represented in the Lax form ${\dot L^{\hbox{\tiny{ZhV}}}(z)}=[L^{\hbox{\tiny{ZhV}}}(z),M^{\hbox{\tiny{ZhV}}}(z)]$
 with $M$-matrix
 \beq\label{w3126}
 \begin{array}{c}
  \displaystyle{
M^{\hbox{\tiny{ZhV}}}(z)=-
\frac{1}{2}\sum\limits_{\al=1}^3 S_\al\,\frac{\vf_1(z)\vf_2(z)\vf_3(z)}{\vf_\al(z)}\,\sigma_{4-\al}\,.
  }
 \end{array}
\eq
 Let $r_{12}(z\pm w)$ be the classical elliptic $r$-matrix \cite{Skl2}:
\beq\label{w331}
  \begin{array}{c}
  \displaystyle{
r_{12}(z\pm w)=\frac{1}{2}\sum\limits_{\al=1}^3 \vf_\al(z\pm w)\sigma_{4-\al}\otimes\sigma_{4-\al}\,.
}
 \end{array}
 \eq
\begin{predl}
The Lax matrix (\ref{w310}) satisfies the classical linear reflection equation
\beq\label{w313}
  \begin{array}{c}
  \displaystyle{
\{L_1^{\hbox{\tiny{ZhV}}}(z),L_2^{\hbox{\tiny{ZhV}}}(w)\}=
\frac{1}{2}\,[L_1^{\hbox{\tiny{ZhV}}}(z)+L_2^{\hbox{\tiny{ZhV}}}(w),r_{12}(z-w)]-
\frac{1}{2}\,[L_1^{\hbox{\tiny{ZhV}}}(z)-L_2^{\hbox{\tiny{ZhV}}}(w),r_{12}(z+w)]
}
 \end{array}
 \eq
identically in $z,w$ and provides the Poisson brackets (\ref{w312}).
\end{predl}

\subsection{Relation of non-relativistic gyrostat to ${\rm BC}_1$ Calogero-Inozemtsev system}
Introduce the following matrix:
\beq\label{w315}
  \begin{array}{c}
  \displaystyle{
\Xi(z)=\mat{\theta_3(z-2q|2\tau)}{-\theta_3(z+2q|2\tau)}{-\theta_2(z-2q|2\tau)}{\theta_2(z+2q|2\tau)}\,.
}
 \end{array}
 \eq
\begin{predl}
The gauge transformation of the Lax matrix (\ref{w240}) (of ${\rm BC}_1$ Calogero-Inozemtsev system) with
the gauge transformation (\ref{w315}) yields the Lax matrix of the Zhukovsky-Volterra (non-relativistic)
gyrostat (\ref{w310})
\beq\label{w316}
  \begin{array}{c}
  \displaystyle{
\Xi(z) L^{\hbox{\tiny{Inoz}}}(z)\Xi^{-1}(z)=L^{\hbox{\tiny{ZhV}}}(z)
}
 \end{array}
 \eq
and provides the Poisson map given by the change of variables
\beq\label{w317}
  \begin{array}{c}
  \displaystyle{
S_\al(p,q)=c_\al\Big(\frac{p}{2}\,\vf_\al(2q)+
\tnu_0\vf_\be(2q)\vf_\ga(2q)+\vf_\al(2q)\sum\limits_{k=1}^3\tnu_k\vf_k(2q)\Big)
}
 \end{array}
 \eq
for distinct $\al,\be,\ga\in\{1,2,3\}$,
and
\beq\label{w318}
  \begin{array}{c}
  \displaystyle{
\la_\al=\frac{\tnu_\al}{c_\al}\,,\quad \al=1,2,3\,,
 }
 \end{array}
 \eq
where $c_\al$ are the constants (\ref{a35}).
The Casimir function $S_1^2+S_2^2+S_3^2$ of the Poisson brackets (\ref{w312}) is mapped under
(\ref{w317}) to
$\tnu_0^2$.
\end{predl}
The statement that the change of variables is a Poisson map means that the brackets (\ref{w312}) are valid for the functions $S_\al(p,q)$
(\ref{w317}), that is
$\{S_\al(p,q),S_\be(p,q)\}=-\imath\varepsilon_{\al\be\ga}S_\ga(p,q)$,
 where $\{S_\al(p,q),S_\be(p,q)\}$ are computed through the canonical Poisson bracket
(\ref{w200}). In fact, this statement was not verified in \cite{LOZ05}. Here we also skip the proof of it since
below a more general statement will be proved for the quadratic Poisson structure.

The matrix (\ref{w315}) is known from the IRF-Vertex correspondence \cite{Baxter-irf}. Its geometrical meaning
and applications to the classical integrable systems can be found in \cite{LOZ}, see also \cite{AtZ1,LOZ2014}.

\subsection{Quadratic Poisson structure: ${\rm BC}_1$ Sklyanin algebra}
Introduce the following quadratic Poisson algebra:
\beq\label{w320}
  \begin{array}{c}
  \displaystyle{
c\{S_\al,S_\be\}=-i\varepsilon_{\al\be\ga} S_0 S_\ga\,,
}
 \end{array}
 \eq
\beq\label{w321}
  \begin{array}{c}
  \displaystyle{
c\{S_0,S_\al\}=-i\varepsilon_{\al\be\ga} S_\be S_\ga(\wp_\be-\wp_\ga)
+i\varepsilon_{\al\be\ga}(S_\be\lambda_\ga-\lambda_\be S_\ga)\,.
}
 \end{array}
 \eq
%
%
The constant $c$ is not necessary, and one can choose $c=1$, but we keep it since
it will be identified with the ''light speed'' $c$ entering (\ref{w234}) through $e^{\pm p/2c}$.
In the case $\lambda_1=\lambda_2=\lambda_3=0$ the brackets (\ref{w320})-(\ref{w321})
define the classical Sklyanin algebra \cite{Skl2}.
The Casimir functions are as follows:
\beq\label{w322}
  \begin{array}{c}
  \displaystyle{
C_1=S_1^2+S_2^2+S_3^2\,,
}
\\ \ \\
  \displaystyle{
C_2=S_0^2+\sum\limits_{k=1}^3 (S_k^2\wp_k+2S_k\lambda_k)\,,\quad \wp_k=\wp(\om_k)\,.
}
 \end{array}
 \eq
Introduce the Lax matrix
\beq\label{w330}
  \begin{array}{c}
  \displaystyle{
L^{\hbox{\tiny{ZhV}}}(z)=\sigma_0S_0
+\sum\limits_{\al=1}^3 \Big(S_\al\vf_\al(z)-\frac{\lambda_\al}{\vf_\al(z)}\Big)\sigma_{4-\al}\,.
}
 \end{array}
 \eq
This Lax matrix satisfies the classical quadratic reflection equation \cite{Skl-refl}.
More precisely, the following statement holds.
\begin{predl}
The Lax matrix (\ref{w330}) satisfies the classical quadratic reflection equation
\beq\label{w332}
  \begin{array}{c}
  \displaystyle{
\{L_1^{\hbox{\tiny{ZhV}}}(z),L_2^{\hbox{\tiny{ZhV}}}(w)\}=
\frac{1}{2c}\,[L_1^{\hbox{\tiny{ZhV}}}(z)L_2^{\hbox{\tiny{ZhV}}}(w),r_{12}(z-w)]+
}
\\ \ \\
  \displaystyle{
+\frac{1}{2c}\,L_2^{\hbox{\tiny{ZhV}}}(w)r_{12}(z+w)L_1^{\hbox{\tiny{ZhV}}}(z)
-\frac{1}{2c}\,L_1^{\hbox{\tiny{ZhV}}}(z)r_{12}(z+w)L_2^{\hbox{\tiny{ZhV}}}(w)
}
 \end{array}
 \eq
identically in $z,w$ and provides the Poisson brackets (\ref{w320})-(\ref{w321}).
\end{predl}
Notice that we keep the same notation for the Lax matrix (\ref{w330}) as it was in (\ref{w310}).
In fact, one can add $\sigma_0 S_0$ to (\ref{w310}) since it is proportional to identity matrix
and does not effect the r.h.s. of (\ref{w313}), and $S_0$ is the Casimir function ($\{S_0,S_\al\}=0$) for the linear brackets
(\ref{w312}).

The above remark emerges from the fact that the
Zhukovsky-Volterra gyrostat is a bihamiltonian model. The Poisson structures (\ref{w312}) and (\ref{w320})-(\ref{w321})
are compatible, that is their linear combination is again the Poisson structure.
The generator $S_0$ is the Casimir function for (\ref{w312}), and it becomes the Hamiltonian for the
quadratic brackets
(\ref{w320})-(\ref{w321}). So that, the equations of motion (\ref{w3125}) and the Lax pair (\ref{w330}), (\ref{w3126})
are the same as in the non-relativistic model up to a constant $c$: $M$-matrix (\ref{w3126}) should be re-scaled as
$M^{\hbox{\tiny{ZhV}}}(z)\to (1/c)M^{\hbox{\tiny{ZhV}}}(z)$ and (\ref{w3125}) becomes $c{\dot S}=[S,\hat\wp(S)]-(1/2)[S,\lambda]$.
The Hamiltonian description is different. In the relativistic case
the equations of motion come from the Poisson brackets (\ref{w321}), and $S_0$ plays the role of the Hamiltonian function.

\section{Gauge equivalence and change of variables}\label{sec4}
\setcounter{equation}{0}

In this Section we describe gauge equivalence between the
4-constant ${\rm BC}_1$ Ruijsenaars-van Diejen model (\ref{w253}) and the relativistic
Zhukovsky-Volterra gyrostat.

\subsection{Gauge transformation}

\begin{theor}\label{th1}
The gauge transformation with the matrix $\Xi(z)$ (\ref{w315}) of 4-constants Lax matrix $L(z)$ (\ref{w234})
provides the Lax matrix of the relativistic Zhukovsky-Volterra gyrostat (\ref{w330})
\beq\label{w351}
  \begin{array}{c}
  \displaystyle{
\Xi(z) L(z)\Xi^{-1}(z)=L^{\hbox{\tiny{ZhV}}}(z)\,,
}
 \end{array}
 \eq
with the following change of variables:
\beq\label{w352}
  \begin{array}{c}
  \displaystyle{
S_0(p,q)=\frac12\Big(\tv(q,\eta)e^{p/2c}+\tv(-q,\eta)e^{-p/2c}\Big)
}
 \end{array}
 \eq
and
\beq\label{w353}
  \begin{array}{c}
   \displaystyle{
S_\al(p,q|\eta,\nu_0,\nu_1,\nu_2,\nu_3)=S_\al(p,q)=
}
\\ \ \\
  \displaystyle{
=\frac{c_\al}{2}\Big(\tv(q,\eta)e^{p/2c}-\tv(-q,\eta)e^{-p/2c}\Big)\vf_\al(2q)+
c_\al\tnu_0\vf_\be(2q)\vf_\ga(2q)+c_\al\vf_\al(2q)\sum\limits_{k=1}^3\tnu_k\vf_k(2q)
}
 \end{array}
 \eq
for distinct $\al,\be,\ga\in\{1,2,3\}$,
where the notation (\ref{w212})
\beq\label{w355}
  \begin{array}{c}
  \displaystyle{
\tv(q,\eta)=\sum\limits_{a=0}^3\tnu_a\vf_a(2q,\eta+\om_a)=v(\eta,q)\,.
 }
 \end{array}
 \eq
is used.
The identification of parameters is given by
\beq\label{w356}
  \begin{array}{c}
  \displaystyle{
\la_\al=\frac{\tnu_\al}{c_\al}\,,\quad \al=1,2,3
 }
 \end{array}
 \eq
with the constants $c_\al$ (\ref{a35}).
\end{theor}
\begin{proof}
The proof is again by direct calculation.
Using (\ref{w213}) let us write $L(z)$ (\ref{w234}) through the dual constants $\tnu_a$ (\ref{w2113}):
\beq\label{w3561}
  \begin{array}{c}
  \displaystyle{
{L}(z)=\mat{\tv(q,\eta)e^{p/2c}}{\tv(q,z)}{\tv(-q,z)}{\tv(-q,\eta)e^{-p/2c}}\,.
}
 \end{array}
 \eq
In order to compute $\Xi(z) L(z)\Xi^{-1}(z)$ one should first use
formulae (\ref{a58})-(\ref{a59}) to transform theta functions with $2\tau$. Then
the Riemann identities for theta-functions should be used. The identities can be found in \cite{Mum},
where the Riemann notations are related to the Jacobi numeration as follows:
$\theta_{11}(x)=-\theta_1(x)$, $\theta_{10}(x)=\theta_2(x)$, $\theta_{00}(x)=\theta_3(x)$ and $\theta_{01}(x)=\theta_4(x)$.

Let us represent (\ref{w3561}) in the form $L(z)=\sum\limits_{\al=0}^{3}\breve\nu_\al L^{\al}(z)$, where
\beq\label{w357}
 \begin{array}{c}
  \displaystyle{
{L^{\al}}(z)=\mat{\vf_{\al}(2q,\eta+\om_\al)e^{p/2c}}{\vf_{\al}(2q,z+\om_\al)}{\vf_{\al}(-2q,z+\om_\al)}
{\vf_{\al}(-2q,\eta+\om_\al)e^{-p/2c}}\,,
}
 \end{array}
\eq
The gauge transformation is performed separately for each $L^a(z)$. Let us demonstrate the calculation  for $L^0(z)$:
\beq\label{w358}
 \begin{array}{c}
  \displaystyle{
(\Xi(z) L^{0}(z)\Xi^{-1}(z))_{11}=\frac{1}{\det\Xi}\Big[\theta_2(z+2q|2\tau)\Big(\theta_3(z-2q|2\tau)\phi(2q,\eta)e^{p/2c}-
}
\\ \ \\
\displaystyle{
-\theta_3(z+2q|2\tau)\phi(-2q,z)\Big)+\theta_2(z-2q|2\tau)\Big(\theta_3(z-2q|2\tau)\phi(2q,z)-
}
\\ \ \\
\displaystyle{
-\theta_3(z+2q|2\tau)\phi(-2q,\eta)e^{-p/2c}\Big)\Big]=
\frac{-1}{2\theta_1(z)\theta_1(2q)}\Big[\phi(2q,\eta)e^{p/2c}\Big(\theta_2(z)\theta_2(2q)-
}
\\ \ \\
\displaystyle{
-\theta_1(z)\theta_1(2q)\Big)-\phi(-2q,\eta)e^{-p/2c}\Big(\theta_2(z)\theta_2(2q)+\theta_1(z)\theta_1(2q)\Big)+
}
\\ \ \\
\displaystyle{
+\theta_2(0)\Big(\theta_2(z-2q)\phi(2q,z)-\theta_2(z+2q)\phi(-2q,z)\Big)\Big]=
}
\\ \ \\
\displaystyle{
=\frac{1}{2}\Big(\phi(2q,\eta)e^{p/2c}+\phi(-2q,\eta)e^{-p/2c}\Big)+
\frac{c_1}{2}\vf_1(2q)\Big(\phi(2q,\eta)e^{p/2c}-\phi(-2q,\eta)e^{-p/2c}\Big)\vf_1(z)+
}
\\ \ \\
\displaystyle{
+c_1\vf_2(2q)\vf_3(2q)\vf_1(z)\,.
}
 \end{array}
\eq
In this calculations we used (\ref{a58}),(\ref{a56}), (\ref{a61}), and $c_k$ are defined in (\ref{a35}).
In the same manner one gets:
\beq\label{w359}
 \begin{array}{c}
  \displaystyle{
(\Xi(z) L^{0}(z)\Xi^{-1}(z))_{12}= \frac{-\imath c_2}{2}\vf_2(2q)\Big(\phi(2q,\eta)e^{p/2c}-\phi(-2q,\eta)e^{-p/2c}\Big)\vf_2(z)-
}
\\ \ \\
\displaystyle{
-\imath c_2\vf_1(2q)\vf_3(2q)\vf_2(z)+\frac{c_3}{2}\vf_3(2q)\Big(\phi(2q,\eta)e^{p/2c}-\phi(-2q,\eta)e^{-p/2c}\Big)\vf_3(z)+
}
\\ \ \\
\displaystyle{
+c_3\vf_1(2q)\vf_2(2q)\vf_3(z)\,,
}
 \end{array}
\eq
\beq\label{w360}
 \begin{array}{c}
  \displaystyle{
(\Xi(z) L^{0}(z)\Xi^{-1}(z))_{21}= \frac{\imath c_2}{2}\vf_2(2q)\Big(\phi(2q,\eta)e^{p/2c}-\phi(-2q,\eta)e^{-p/2c}\Big)\vf_2(z)+
}
\\ \ \\
\displaystyle{
+\imath c_2\vf_1(2q)\vf_3(2q)\vf_2(z)+\frac{c_3}{2}\vf_3(2q)\Big(\phi(2q,\eta)e^{p/2c}-\phi(-2q,\eta)e^{-p/2c}\Big)\vf_3(z)+
}
\\ \ \\
\displaystyle{
+c_3\vf_1(2q)\vf_2(2q)\vf_3(z)\,,
}
 \end{array}
\eq
\beq\label{w361}
 \begin{array}{c}
  \displaystyle{
(\Xi(z) L^{0}(z)\Xi^{-1}(z))_{22}=\frac{1}{2}\Big(\phi(2q,\eta)e^{p/2c}+\phi(-2q,\eta)e^{-p/2c}\Big)-
}
\\ \ \\
\displaystyle{
-\frac{c_1}{2}\vf_1(2q)\Big(\phi(2q,\eta)e^{p/2c}-\phi(-2q,\eta)e^{-p/2c}\Big)\vf_1(z)-c_1\vf_2(2q)\vf_3(2q)\vf_1(z)\,.
}
 \end{array}
\eq
This finishes the proof.
\end{proof}

\subsection{Poisson map}
We begin this subsection with a lemma formulating an answer
for the expression $\tv(q,\eta)\tv(-q,\eta)$ (or equivalently, $v(\eta,q)v(\eta,-q)$),
which is
alternative to (\ref{a50}), see (\ref{a551}).
\begin{lemma}\label{lem1}
The following identity holds true for the function $v$ (\ref{w210}):
\beq\label{w370}
  \begin{array}{c}
  \displaystyle{
v(\eta,q)v(\eta,-q)=
}
\\
  \displaystyle{
=\sum\limits_{k=0}^3\tnu_k^2\Big(\wp(\eta+\om_k)-\wp(2q)\Big)
-2\tnu_0\Big(\tnu_1\vf_2(2q)\vf_3(2q)+\tnu_2\vf_1(2q)\vf_3(2q)+\tnu_3\vf_1(2q)\vf_2(2q)\Big)-
}
\\
  \displaystyle{
  -2\Big(\tnu_1\tnu_2\vf_1(2q)\vf_2(2q)+\tnu_2\tnu_3\vf_2(2q)\vf_3(2q)+\tnu_1\tnu_3\vf_1(2q)\vf_3(2q)\Big)\,.
 }
 \end{array}
 \eq
\end{lemma}
\begin{proof}
Using (\ref{a19}) and (\ref{a223}) we have
\beq\label{w371}
  \begin{array}{c}
  \displaystyle{
v(\eta,q)v(\eta,-q)=\tv(q,\eta)\tv(-q,\eta)=
}
\\ \ \\
  \displaystyle{
=\sum\limits_{k=0}^3\tnu_k^2\vf_k(2q,\eta+\om_k)\vf_k(-2q,\eta+\om_k)+
\sum\limits_{k\neq m}^3\tnu_k\tnu_m\vf_k(2q,\eta+\om_k)\vf_m(-2q,\eta+\om_m)=
}
\\
  \displaystyle{
=\sum\limits_{k=0}^3\tnu_k^2\Big(\wp(\eta+\om_k)-\wp(2q)\Big)
-
}
\\
  \displaystyle{
-\sum\limits_{k<m}^3\tnu_k\tnu_m\Big(\vf_k(2q,\eta+\om_k)\vf_m(2q,-\eta+\om_m)
+\vf_m(2q,\eta+\om_m)\vf_k(2q,-\eta+\om_k)\Big)\,.
 }
 \end{array}
 \eq
Consider expression from the last sum\footnote{Similar calculation was performed in \cite{Z04}.}.
For any distinct $\al,\be,\ga\in\{1,2,3\}$ due to (\ref{a17}), (\ref{a06}) and (\ref{a13})-(\ref{a141})
\beq\label{w372}
  \begin{array}{c}
  \displaystyle{
\vf_\al(2q,\eta+\om_\al)\vf_\be(2q,-\eta+\om_\be)
+\vf_\be(2q,\eta+\om_\be)\vf_\al(2q,-\eta+\om_\al)=
 }
 \\ \ \\
  \displaystyle{
=\vf_\ga(2q,\om_\ga)\Big(E_1(2q)+E_1(\eta+\om_\al)+E_1(-\eta+\om_\be)-E_1(2q+\om_\al+\om_\be)\Big)+
 }
  \\ \ \\
  \displaystyle{
+\vf_\ga(2q,\om_\ga)\Big(E_1(2q)+E_1(\eta+\om_\be)+E_1(-\eta+\om_\al)-E_1(2q+\om_\al+\om_\be)\Big)=
 }
   \\ \ \\
  \displaystyle{
=2\vf_\ga(2q,\om_\ga)\Big(E_1(2q)+E_1(\om_\al+\om_\be)-E_1(2q+\om_\al+\om_\be)\Big)=2\vf_\al(2q)\vf_\be(2q)\,.
 }
 \end{array}
 \eq
Similar calculation yields
\beq\label{w373}
  \begin{array}{c}
  \displaystyle{
\vf_\ga(2q,\eta+\om_\ga)\phi(2q,-\eta)
+\phi(2q,\eta)\vf_\ga(2q,-\eta+\om_\ga)=
2\vf_\al(2q)\vf_\be(2q)\,.
 }
 \end{array}
 \eq
Plugging (\ref{w372}) and (\ref{w373}) into (\ref{w371}) on gets (\ref{w370}).
\end{proof}
\begin{theor}\label{th2}
The change of variables from $p,q$ to $S_a(p,q)$, $a=0,...,3$ given by (\ref{w352})-(\ref{w353})
and identification of parameters (\ref{w356})
provides the Poisson map between the canonical brackets (\ref{w200}) and the ${\rm BC}_1$ Sklyanin
algebra (\ref{w320})-(\ref{w321}). The Casimir functions $C_1$ and $C_2$ (\ref{w322}) under this change of variables
are mapped to the values
\beq\label{w3731}
  \begin{array}{c}
  \displaystyle{
C_1=\tnu_0^2\,,\quad C_2=\sum\limits_{a=0}^3\tnu_a^2\wp(\eta+\om_a)-\sum\limits_{a=1}^3\tnu_a^2\wp(\om_a)
 }
 \end{array}
 \eq
\end{theor}
\begin{proof}
The statement of the theorem means that the Poisson brackets between the functions $S_a(p,q)$ (\ref{w352})-(\ref{w353})
computed through the canonical Poisson structure (\ref{w200}) yields  (\ref{w320})-(\ref{w321}).

Below we assume $\al,\be,\ga\in\{1,2,3\}$ are distinct. Consider first $\{S_\al,S_\be\}$:
\beq\label{w374}
  \begin{array}{c}
  \displaystyle{
\frac{c}{c_\al c_\be}\{S_\al,S_\be\}=
c\frac{\p_p S_\al}{c_\al}
\frac{\p_q S_\be}{c_\be}
-
c\frac{\p_q S_\al}{c_\al}
\frac{\p_p S_\be}{c_\be}\,.
 }
 \end{array}
 \eq
Obviously,
\beq\label{w375}
  \begin{array}{c}
  \displaystyle{
2c\frac{\p_p S_\al}{c_\al}=S_0\vf_\al(2q)
 }
 \end{array}
 \eq
and using (\ref{a32}) we have
\beq\label{w376}
  \begin{array}{c}
  \displaystyle{
\frac{1}{2}\frac{\p_q S_\be}{c_\be}=\frac{1}{4}\Big( v'(\eta,q)e^{p/2c}+v'(\eta,-q)e^{-p/2c}\Big)\vf_\be(2q)
-
 }
 \\ \ \\
   \displaystyle{
 -\frac{1}{2}\Big( v(\eta,q)e^{p/2c}-v(\eta,-q)e^{-p/2c}\Big)\vf_\al(2q)\vf_\ga(2q)
 -\tnu_0\Big(\vf_\al^2(2q)+\vf_\ga^2(2q)\Big)\vf_\be(2q)-
 }
  \\ \ \\
   \displaystyle{
 -\tnu_\al\Big(\vf_\al^2(2q)+\vf_\be^2(2q)\Big)\vf_\ga(2q)-2\tnu_\be\vf_\al(2q)\vf_\be(2q)\vf_\ga(2q)
  -\tnu_\ga\Big(\vf_\be^2(2q)+\vf_\ga^2(2q)\Big)\vf_\al(2q)\,.
 }
 \end{array}
 \eq
Similar expressions are obtained for $\p_p S_\be$ and $\p_q S_\al$. Plugging it into (\ref{w374})
and using (\ref{a31}) we get (the terms with $v'(\eta,\pm q)$ are cancelled out):
\beq\label{w377}
  \begin{array}{c}
  \displaystyle{
\frac{c}{c_\al c_\be}\{S_\al,S_\be\}=S_0 S_\ga\frac{\wp_\al-\wp_\be}{c_\ga}\,.
 }
 \end{array}
 \eq
Then the Poisson brackets (\ref{w320}) follow from the identity (\ref{a38}).

The computation of $\{S_0,S_\al\}$ is performed similarly but this time the terms with $v'(\eta,\pm q)$
are not cancelled out. These type terms appear in the combination $v(\eta,q)v'(\eta,-q)-v'(\eta,q)v(\eta,-q)$.
Instead of using (\ref{a54}), here one should differentiate (\ref{w370}) with respect to $q$.
 This yields the following identity:
\beq\label{w378}
  \begin{array}{c}
  \displaystyle{
-\frac14\Big(v(\eta,q)v'(\eta,-q)-v'(\eta,q)v(\eta,-q)\Big)=

 }
 \\ \ \\
  \displaystyle{
  =(\tnu_0^2+\tnu_\al^2+\tnu_\be^2+\tnu_\ga^2)\vf_\al(2q)\vf_\be(2q)\vf_\ga(2q)+
(\tnu_0\tnu_\al+\tnu_\be\tnu_\ga)\Big(\vf_\be^2(2q)+\vf_\ga^2(2q)\Big)\vf_\al(2q)+

 }
  \\ \ \\
  \displaystyle{
  +(\tnu_0\tnu_\be+\tnu_\al\tnu_\ga)\Big(\vf_\al^2(2q)+\vf_\ga^2(2q)\Big)\vf_\be(2q)
+(\tnu_0\tnu_\ga+\tnu_\al\tnu_\be)\Big(\vf_\al^2(2q)+\vf_\be^2(2q)\Big)\vf_\ga(2q)\,,
 }
 \end{array}
 \eq
where we also used (\ref{a33}).
Using (\ref{w378}) together with (\ref{a31}) one can verify that
\beq\label{w379}
  \begin{array}{c}
  \displaystyle{
\frac{c}{c_\al}\{S_0,S_\al\}=-\frac{1}{c_\be c_\ga}\,S_\be S_\ga
+\frac{\wp_\al-\wp_\ga}{c_\be}\,S_\be\tnu_\ga
+\frac{\wp_\al-\wp_\be}{c_\ga}\,S_\ga\tnu_\be\,.
 }
 \end{array}
 \eq
Then, with the help of (\ref{a38}) one obtains (\ref{w321}) with $\la_\al=\tnu_\al/c_\al$ as in (\ref{w356}).

Finally, in order to get (\ref{w3731}) we notice that $\det L(z)=\det L^{\hbox{\tiny{ZhV}}}(z)$ due to (\ref{w351}),
and the Casimir functions (\ref{w322}) are generated by the determinant
\beq\label{w380}
  \begin{array}{c}
  \displaystyle{
\det L^{\hbox{\tiny{ZhV}}}(z)=-\wp(z)C_1+C_2-\sum\limits_{k=1}^3\frac{\lambda_k^2}{\vf_k^2(z)}\,.
 }
 \end{array}
 \eq
On the other hand
\beq\label{w381}
  \begin{array}{c}
  \displaystyle{
\det L(z)=v(\eta,q)v(\eta,-q)-v(z,q)v(z,-q)=\sum\limits_{a=0}^3\tnu_a^2\Big(\wp(\eta+\om_a)-\wp(z+\om_a)\Big)\,.
 }
 \end{array}
 \eq
By comparing (\ref{w380}) and (\ref{w381}) and using (\ref{a261}) one gets (\ref{w3731}).
\end{proof}

\subsection{Coupled Zhukovsky-Volterra gyrostats}
It is now straightforward to perform the gauge transformation (\ref{w351}) to the
8-constant model (\ref{w233})-(\ref{w235}). Indeed,
\beq\label{w382}
  \begin{array}{c}
  \displaystyle{
\Xi(z) \mL(z)\Xi^{-1}(z)=\Xi(z) L(z)\Xi^{-1}(z)\Xi(z) {\bar L}(z)\Xi^{-1}(z)=L^{\hbox{\tiny{ZhV}}}(z)
{\bar L}^{\hbox{\tiny{ZhV}}}(z)\,,
 }
 \end{array}
 \eq
where the matrix $L^{\hbox{\tiny{ZhV}}}(z)$ is given by the expression (\ref{w330}) and
\beq\label{w383}
  \begin{array}{c}
  \displaystyle{
{\bar L}^{\hbox{\tiny{ZhV}}}(z)=\sigma_0{\bar S}_0
+\sum\limits_{\al=1}^3 \Big({\bar S}_\al\vf_\al(z)-\frac{{\bar\lambda}_\al}{\vf_\al(z)}\Big)\sigma_{4-\al}\,.
 }
 \end{array}
 \eq
Expressions for ${\bar S}_\al(p,q|\bar\eta,\bnu_a)$ are the same as for $S_\al(p,q|\eta,\nu_a)$
but with another set of constants. That is, ${\bar S}_\al(p,q)=S_\al(p,q|\bar\eta,\bnu_a)$.
The Poisson brackets $\{S_a,S_b\}$ and $\{{\bar S}_a,{\bar S}_b\}$ have the form (\ref{w320})
and
\beq\label{w384}
  \begin{array}{c}
  \displaystyle{
c\{{\bar S}_\al,{\bar S}_\be\}=-i\varepsilon_{\al\be\ga} {\bar S}_0 {\bar S}_\ga\,,
}
\\ \ \\
  \displaystyle{
c\{{\bar S}_0,{\bar S}_\al\}=-i\varepsilon_{\al\be\ga} {\bar S}_\be {\bar S}_\ga(\wp_\be-\wp_\ga)
+i\varepsilon_{\al\be\ga}({\bar S}_\be{\bar\lambda}_\ga-{\bar\lambda}_\be {\bar S}_\ga)\,.
}
 \end{array}
 \eq
The mixed type brackets $\{{S}_a,{\bar S}_b\}$ can be computed as function of $p,q$
but an answer in terms of $S_a$, ${\bar S}_b$ is unknown.

\section{1-site XYZ model with boundaries}\label{sec5}
\setcounter{equation}{0}
Here we compute the transfer-matrix for the 1-site classical XYZ chain with boundaries
and show how it is related to the Ruijsenaars-van Diejen model (\ref{w01}).

\subsection{$L$-matrix and $K$-matrices}
\paragraph{Original classical Sklyanin algebra \cite{Skl2}.}
We consider the XYZ Lax matrix
\beq\label{w525}
\mathbb{L}(z) = \mathbb{S}_0\sigma_0+\sum\limits_{\al=1}^3 \mathbb{S}_{\al}\vf_{\al}(z)\sigma_{4-\al}
\eq
with the Sklyanin algebra generators $\mathbb{S}_0,\mathbb{S}_{\al}$ (without additional constants).
The Poisson brackets
\beq\label{w5252}
  \begin{array}{c}
  \displaystyle{
c\{\mathbb{S}_\al,\mathbb{S}_\be\}=-i\varepsilon_{\al\be\ga} \mathbb{S}_0 \mathbb{S}_\ga\,,
}
\\ \ \\
  \displaystyle{
c\{\mathbb{S}_0,\mathbb{S}_\al\}=
-i\varepsilon_{\al\be\ga} \mathbb{S}_\be \mathbb{S}_\ga\big(\wp(\om_\be)-\wp(\om_\ga)\big)
}
 \end{array}
 \eq
are generated by the classical quadratic $r$-matrix structure (with $r$-matrix (\ref{w331}))
\beq\label{w5251}
  \begin{array}{c}
  \displaystyle{
\{\mathbb{L}_1(z),\mathbb{L}_2(w)\}=\frac{1}{c}\,[\mathbb{L}_1(z)\mathbb{L}_2(w),r(z-w)]
 }
 \end{array}
 \eq
with the Casimir functions
\beq\label{w527}
  \begin{array}{c}
  \displaystyle{
{\bf C}_1=\mathbb{S}_1^2+\mathbb{S}_2^2+\mathbb{S}_3^2\,,
\qquad
{\bf C}_2=\mathbb{S}_0^2+\sum\limits_{k=1}^3 \mathbb{S}_k^2\wp(\om_k)\,.
}
 \end{array}
 \eq
From viewpoint of dynamical systems in classical mechanics, the Lax matrix (\ref{w525}) defines the (relativistic) Euler top
with the equations of motion
\beq\label{w5254}
  \begin{array}{c}
  \displaystyle{
 c{\dot {\mathbb S}}=[{\mathbb S},\hat\wp({\mathbb S})]\,,
 \qquad
 {\mathbb S}=\sum\limits_{\al=1}^3\sigma_{4-\alpha}{\mathbb S}_\al\,,
 \qquad
 \hat\wp({\mathbb S})=-\frac12\sum\limits_{\al=1}^3\sigma_{4-\alpha}{\mathbb S}_\al\wp(\om_\al)\,,
 }
 \end{array}
 \eq
which follow from the brackets (\ref{w5252}) and the Hamiltonian ${\mathbb S}_0$. The $M$-matrix (\ref{w3126})
is almost the same here since it is independent of additional constants:
 ${\mathbb M}(z)=(1/c)M^{\hbox{\tiny{ZhV}}}|_{S_a={\mathbb S}_a}$.

Next, consider the classical analogues of representation of quantum Sklyanin's generators in terms difference operators:
\beq\label{w526}
\mathbb{S}_{a}=\mathbb{S}_{a}(p,q,\eta)=
d_{a}\Big(\vf_{a}(2q,\eta+\om_{a})e^{p/2c}+\vf_{a}(-2q,\eta+\om_{a})e^{-p/2c}\Big)\,,
\eq
for $a=0,1,2,3$, and
\beq\label{w5261}
d_{0}=1\,,\qquad d_\al=c_\al\vf_\al(\eta)\,,\quad \al=1,2,3\,,
\eq
where $c_\al$ are constants from (\ref{a35})-(\ref{a36}).
One can deduce these formulae by applying the gauge transformation (\ref{w315})
to the Lax matrix of the ${\rm A}_1$ Ruijsenaars-Schneider model (\ref{w252}) in the center of mass frame (with certain redefinition of constants in (\ref{w252}) including $c\to 2c$, $\eta\to -\eta$):
 ${\mathbb L}(z)|_{\mathbb S=\mathbb S(p,q,\eta)}={\rm const}\cdot\Xi(z)L^{\hbox{\tiny{RS}}}(z)\Xi^{-1}(z)$.
%
%
See \cite{ZZ} and references therein for details.

\paragraph{$K$-matrices.} The classical constant non-dynamical boundary $K$-matrices are solutions
to the following reflection equation \cite{Skl-refl}:
\beq\label{w530}
\displaystyle{
[K_1(z)K_2(w),r _{12}(z-w)]+\,K_2(w)r_{12}(z+w)K_1(z)
-\,K_1(z)r_{12}(z+w)K_2(w)=0\,.
}
\eq
Consider its solutions \cite{IK}:
\beq\label{w529}
  \begin{array}{c}
  \displaystyle{
K^{\pm}(z)=\sigma_0+\sum\limits_{\al=1}^3 \rho^{\pm}_\al\vf_{\al}(z+\om_{\al})\sigma_{4-\al}
 }
 \end{array}
\eq
with an arbitrary constants $\rho^{\pm}_\al\in\mC$.
Due to (\ref{a361}) it is equivalently written as
\beq\label{w561}
  \begin{array}{c}
  \displaystyle{
K^{\pm}(z)=\sigma_0+\sum\limits_{\al=1}^3 {\tilde\rho}^{\pm}_\al\,\frac{1}{\vf_{\al}(z)}\,\sigma_{4-\al}\,.
 }
 \end{array}
\eq
Being written in this form it is easy to see from the description of the Zhukovsky-Volterra gyrostat (\ref{w320})-(\ref{w332}) that $K^{\pm}(z)$ indeed solve (\ref{w530}).
Put $S_0=1$ and $S_\al=0$, $\al=1,2,3$.
 Then the Poisson brackets (\ref{w320})-(\ref{w321}) are fulfilled as $0=0$, and the equation (\ref{w332})
turns into (\ref{w530}).

\subsection{Transfer-matrix}
Following the general construction from \cite{Skl-refl} consider the 1-site XYZ chain
with boundaries given by $K^\pm(z)$-matrices (\ref{w561}). Its classical transfer-matrix has the form
\beq\label{w5311}
  \begin{array}{c}
  \displaystyle{
\tr\Big(K^+(z){\mathbb L}(z)K^-(z){\mathbb L}^{-1}(-z)\Big)\,.
 }
 \end{array}
\eq
It is the generating function of integrals of motion.
In our case ${\mathbb L}^{-1}(-z) = \frac{1}{\det {\mathbb L}(z)}{\mathbb L}(z)$ due to (\ref{a28}).
The expression $\det {\mathbb L}(z)$ is a non-dynamical factor since it belongs to the center of algebra
of (\ref{w5251}). Therefore, we may slightly redefine (\ref{w5311}) as
\beq\label{w531}
  \begin{array}{c}
  \displaystyle{
t(z)=\frac12\,\tr\Big(K^+(z){\mathbb L}(z)K^-(z){\mathbb L}(z)\Big)\,.
 }
 \end{array}
\eq
From (\ref{w525}) and (\ref{w561}) we have
\beq\label{w532}
\begin{array}{c}
\displaystyle{
K^{\pm}(z){\mathbb L}(z) = \mathbb{S}_0\sigma_0+
\sum\limits_{\al=1}^3\frac{{\mathbb{S}_0\,\tilde\rho}^{\pm}_{\al}}{\vf_\al(z)}\,\sigma_{4-\al}+
\sum\limits_{\al=1}^3 \mathbb{S}_{\al}\vf_\al(z)\sigma_{4-\al}
+\sum\limits_{\al,\be=1}^3{\tilde\rho}^{\pm}_\al\mathbb{S}_\be\,\frac{\vf_{\be}(z)}{\vf_\al(z)}\, \sigma_{4-\al}\sigma_{4-\be}\,.
}
\end{array}
\eq
Then direct calculation yields
\beq\label{w534}
\begin{array}{c}
\displaystyle{
t(z)=2\Big(\mathbb{S}_0+\sum\limits_{\al=1}^3{\tilde\rho}^+_\al\mathbb{S}_\al\Big)
\Big(\mathbb{S}_0+\sum\limits_{\al=1}^3{\tilde\rho}^-_\al\mathbb{S}_\al\Big)
-\Big(1-\sum\limits_{\al=1}^3\frac{{\tilde\rho}^+_\al{\tilde\rho}^-_\al}{\vf_\al^2(z)}\Big)
\Big({\bf C}_2-{\bf C}_1\wp(z)\Big)\,.
}
\end{array}
\eq
Therefore, the Hamiltonian of the chain with boundaries takes the form
\beq\label{w535}
\begin{array}{c}
\displaystyle{
\mH^{\hbox{\tiny{XYZ}}}=\Big(\mathbb{S}_0+\sum\limits_{\al=1}^3{\tilde\rho}^+_\al\mathbb{S}_\al\Big)
\Big(\mathbb{S}_0+\sum\limits_{\al=1}^3{\tilde\rho}^-_\al\mathbb{S}_\al\Big)\,.
}
\end{array}
\eq

\subsection{Relation to Ruijsenaars-van Diejen model}
Recall that the Hamiltonian of ${\rm BC}_1$ Ruijsenaars-van Diejen model (\ref{w01}) can be written
(up to constant)
in the factorized form ${H}^{\hbox{\tiny{8vD}}}={H}_1^{\hbox{\tiny{4vD}}}{\bar H}_1^{\hbox{\tiny{4vD}}}$.
Consider a special case
\beq\label{w5700}
\begin{array}{c}
\displaystyle{
\eta=\bar\eta
}
\end{array}
\eq
and let us write down both 4-constant
Hamiltonians ${H}_1^{\hbox{\tiny{4vD}}}$ and ${\bar H}_1^{\hbox{\tiny{4vD}}}$
through the dual function $\breve{v}$:
\beq\label{w570}
\begin{array}{c}
\displaystyle{
H^{\hbox{\tiny{4vD}}}_1
=\breve{v}(q,\eta)e^{p/2c}+{\breve v}(-q,\eta)e^{-p/2c}\,,
}
\\ \ \\
\displaystyle{
{\bar H}^{\hbox{\tiny{4vD}}}_1
=\breve{\bv}(q,\eta)e^{p/2c}+\breve{\bv}(-q,\eta)e^{-p/2c}\,.
}
\end{array}
\eq
Notice that these expressions are linear combinations of the Sklyanin's generators ${\mathbb S}_a(p,q,\eta)$:
\beq\label{w572}
\begin{array}{c}
\displaystyle{
H^{\hbox{\tiny{4vD}}}_1=\sum\limits_{a=0}^3\breve{\nu}_a\Big(\vf_a(2q,\eta+\om_a)e^{p/2c}+\vf_a(-2q,\eta+\om_a)e^{-p/2c}\Big)
=\sum\limits_{a=0}^3\frac{\breve{\nu}_a}{d_a}\,{\mathbb S}_a(p,q,\eta)\,,
}
\\
\displaystyle{
{\bar H}^{\hbox{\tiny{4vD}}}_1=\sum\limits_{a=0}^3\breve{\bnu}_a
\Big(\vf_a(2q,\eta+\om_a)e^{p/2c}+\vf_a(-2q,\eta+\om_a)e^{-p/2c}\Big)
=\sum\limits_{a=0}^3\frac{\breve{\bnu}_a}{d_a}\,{\mathbb S}_a(p,q,\eta)\,.
}
\end{array}
\eq
On the other hand, for the Hamiltonian (\ref{w535}) we have (assuming $\breve{\nu}_0,\breve{\bnu}_0\neq 0$):
\beq\label{w573}
\begin{array}{c}
\displaystyle{
\frac{\breve{\nu}_0\breve{\bnu}_0}{d_0^2}\,\mH^{\hbox{\tiny{XYZ}}}=\Big(\frac{\breve{\nu}_0}{d_0}\,\mathbb{S}_0
+\sum\limits_{\al=1}^3\frac{\breve{\nu}_0}{d_0}\,{\tilde\rho}^+_\al\mathbb{S}_\al\Big)
\Big(\frac{\breve{\bnu}_0}{d_0}\,\mathbb{S}_0
+\sum\limits_{\al=1}^3\frac{\breve{\bnu}_0}{d_0}\,{\tilde\rho}^-_\al\mathbb{S}_\al\Big)\,.
}
\end{array}
\eq
This expression coincides with the Hamiltonian ${H}_1^{\hbox{\tiny{4vD}}}{\bar H}_1^{\hbox{\tiny{4vD}}}$
for the following choice of parameters ${\tilde\rho}^\pm_\al$ from $K$-matrices:
\beq\label{w574}
\begin{array}{c}
\displaystyle{
{\tilde\rho}^+_\al=\frac{d_0}{\breve{\nu}_0}\frac{\breve{\nu}_\al}{d_\al}\,,
\qquad
{\tilde\rho}^-_\al=\frac{d_0}{\breve{\bnu}_0}\frac{\breve{\bnu}_\al}{d_\al}\,,
\qquad
\al=1,2,3\,.
}
\end{array}
\eq

\section{Representation of the Chalykh's Lax matrix through the Sklyanin generators}\label{sec6}
\setcounter{equation}{0}
In this Section we use the Sklyanin generators written in the form
\beq\label{w52611}
{\bf S}_{a}=\frac12 \Big(\vf_{a}(2q,\eta+\om_{a})e^{p/2c}+\vf_{a}(-2q,\eta+\om_{a})e^{-p/2c}\Big)\,,
\qquad a=0,1,2,3\,.
\eq
They are simply related to the original Sklyanin's generators $\mathbb{S}_a$ (\ref{w526}), (\ref{w5252}):
\beq\label{w52612}
{\bf S}_{a}=\frac{1}{2d_a}\,{\mathbb S}_a\,,
\qquad a=0,1,2,3\,.
\eq
Using identities (\ref{a63}), (\ref{a64}), (\ref{w372}), (\ref{w373}) and (\ref{a17}), (\ref{a32}), (\ref{a223}), (\ref{a13}), (\ref{a14}) one can find the Poisson brackets for ${\bf S}_{a}$ (\ref{w52611}):
\beq\label{w52521}
  \begin{array}{c}
  \displaystyle{
c\{{\bf S_\al},{\bf S_\be }\}=-\big( I_\al-I_\be\big) {\bf S}_0{\bf S_\ga }\,,
}
\\ \ \\
  \displaystyle{
c\{{\bf S}_0,{\bf S_\al}\}=I_\al
{\bf S_\be} {\bf S_\ga}
}
 \end{array}
 \eq
with
\beq\label{w52527}
  \begin{array}{c}
  \displaystyle{
I_\al=E_1(\eta+\om_\al)- E_1(\om_\al)- E_1(\eta)\,,\quad \al=1,2,3\,.
}
 \end{array}
 \eq
Consider the Chalykh's Lax matrix (\ref{w230}) with $\bar\eta=\eta$, that is, for 7 independent constants case.
Instead of performing the gauge transformation (\ref{w232}) here we just represent it in the
form
\beq\label{w23311}
  \begin{array}{c}
  \displaystyle{
{\mL}^{\hbox{\tiny{Ch}}}(z)={\bf L}(z){\bf {\bar L}}(z)\,,
}
 \end{array}
 \eq
where
\beq\label{w23411}
  \begin{array}{c}
  \displaystyle{
{\bf L}(z)=\mat{v(\eta,q)e^{p/2c}}{-v(z,q)e^{-p/2c}}{-v(z,-q)e^{p/2c}}{v(\eta,-q)e^{-p/2c}}
=\mat{\tv(q,\eta)e^{p/2c}}{-\tv(q,z)e^{-p/2c}}{-\tv(-q,z)e^{p/2c}}{\tv(-q,\eta)e^{-p/2c}}
}
 \end{array}
 \eq
and
\beq\label{w23511}
  \begin{array}{c}
  \displaystyle{
{\bf {\bar L}}(z)=
\mat{\bv(\eta,q)e^{p/2c}}{-\bv(z,q)e^{-p/2c}}{-\bv(z,-q)e^{p/2c}}{\bv(\eta,-q)e^{-p/2c}}
=\mat{\tbv(q,\eta)e^{p/2c}}{-\tbv(q,z)e^{-p/2c}}{-\tbv(-q,z)e^{p/2c}}{\tbv(-q,\eta)e^{-p/2c}}\,.
}
 \end{array}
 \eq
Next, perform the gauge transformation
 \beq\label{w3151}
  \begin{array}{c}
  \displaystyle{
{\mL}^{\hbox{\tiny{Ch}}}(z)\to {\mL}'(z)=\Xi_\eta(z){\mL}^{\hbox{\tiny{Ch}}}\Xi^{-1}_\eta(z)=
\Xi_\eta(z) {\bf L}(z)\Xi_\eta(z)\cdot\Xi^{-1}_\eta(z){\bf {\bar L}}(z)\Xi^{-1}_\eta(z)
}
 \end{array}
 \eq
with the following matrix:
 \beq\label{w31513}
  \begin{array}{c}
  \displaystyle{
\Xi_\eta(z)=
\mat{\theta_3(z+\eta-2q|2\tau)}{-\theta_3(z+\eta+2q|2\tau)}{-\theta_2(z+\eta-2q|2\tau)}{\theta_2(z+\eta+2q|2\tau)}\,.
}
 \end{array}
 \eq
 Denote
 \beq\label{w31514}
  \begin{array}{c}
  \displaystyle{
{\bf L}'(z)=\Xi_\eta(z) {\bf L}(z)\Xi^{-1}_\eta(z)\,,\qquad
 {\bf {\bar L}}'(z)= \Xi_\eta(z){\bf {\bar L}}(z)\Xi^{-1}_\eta(z)\,,
}
 \end{array}
 \eq
 that is ${\mL}'(z)={\bf L}'(z) {\bf {\bar L}}'(z)$.

\begin{theor}\label{th3}
Dependence of the matrix elements of the (gauge transformed) matrix ${\bf L}'(z)$ on dynamical variables
 is expressed through the Sklyanin's generators ${\bf S}_a$ (\ref{w52611}) only.
\end{theor}
\begin{proof}
Direct calculation gives
\beq\label{w35822}
 \begin{array}{c}
  \displaystyle{
{\bf L}'(z)_{11}=
\frac{-1}{2\theta_1(z+\eta)\theta_1(2q)}\Big[\tv(q,\eta)e^{p/2c}\Big(\theta_2(z+\eta)\theta_2(2q)-
}
\\ \ \\
\displaystyle{
-\theta_1(z+\eta)\theta_1(2q)\Big)-\tv(-q,\eta)e^{-p/2c}\Big(\theta_2(z+\eta)\theta_2(2q)+\theta_1(z+\eta)\theta_1(2q)\Big)-
}
\\ \ \\
\displaystyle{
-\theta_2(0)\Big(\theta_2(z+\eta-2q)\tv(q,z)e^{-p/2c}-\theta_2(z+\eta+2q)\tv(-q,z)e^{p/2c}\Big)\Big]=
}
\\ \ \\
\displaystyle{
= \frac{-1}{2\theta_1(z+\eta)\theta_1(2q)}\sum\limits_{\al=0}^{3}\breve{\nu}_\al e^{p/2c}\frac{\theta'_1(0)}{\theta_1(2q)\theta_{\al+1}(\eta)\theta_{\al+1}(z)}\Big[\theta_{\al+1}(\eta+2q)\theta_{\al+1}(z)\theta_{2}(z+\eta)\theta_{2}(2q)-
}
\\ \ \\
\displaystyle{
-\theta_{\al+1}(z-2q)\theta_{\al+1}(\eta)\theta_{2}(\eta+z+2q)\theta_{2}(0)\Big] + (q\rightarrow{-q}, p\rightarrow{-p}) +
}
\\ \ \\
\displaystyle{
 + \frac12\tv(q,\eta)e^{p/2c}+\frac12\tv(-q,\eta)e^{-p/2c}\,.
}
\end{array}
\eq
The expression in the last square brackets can be simplified by using the Weierstrass identities (\ref{w623})-(\ref{w627}). Plugging $u=y=\frac12
(\eta+z+2q)$, $x=\frac12
(\eta-z+2q)$, $v=\frac12
(\eta+z-2q)$ into (\ref{w623})-(\ref{w627}) we obtain
\beq\label{w358221}
 \begin{array}{c}
  \displaystyle{
{\bf L}'(z)_{11}= \sum\limits_{\al=0}^{3}\breve{\nu}_{\al}{\bf S_\al}-\breve{\nu}_0{\bf S_1}\frac{\theta_2(\eta)}{\theta_1(\eta)} \frac{\theta_2(z)}{\theta_1(z)}-\breve{\nu}_1{\bf S_0}\frac{\theta_1(\eta)}{\theta_2(\eta)} \frac{\theta_1(z)}{\theta_2(z)}-
}
\\ \ \\
\displaystyle{
-\breve{\nu}_2{\bf S_3}\frac{\theta_4(\eta)}{\theta_3(\eta)} \frac{\theta_4(z)}{\theta_3(z)}-\breve{\nu}_3{\bf S_2}\frac{\theta_3(\eta)}{\theta_4(\eta)} \frac{\theta_3(z)}{\theta_4(z)}\,.
}
 \end{array}
\eq
The evenness of $\theta_a(z)$, $a \in\{2,3,4\}$ was also taken into account.
Other matrix elements are calculated in a similar fashion:
\beq\label{w358222}
  \begin{array}{c}
  \displaystyle{
{\bf L}'(z)_{22}= \sum\limits_{\al=0}^{3}\breve{\nu}_{\al}{\bf S_\al}+\breve{\nu}_0{\bf S_1}\frac{\theta_2(\eta)}{\theta_1(\eta)} \frac{\theta_2(z)}{\theta_1(z)}+\breve{\nu}_1{\bf S_0}\frac{\theta_1(\eta)}{\theta_2(\eta)} \frac{\theta_1(z)}{\theta_2(z)}+
}
\\ \ \\
\displaystyle{
+\breve{\nu}_2{\bf S_3}\frac{\theta_4(\eta)}{\theta_3(\eta)} \frac{\theta_4(z)}{\theta_3(z)}+\breve{\nu}_3{\bf S_2}\frac{\theta_3(\eta)}{\theta_4(\eta)} \frac{\theta_3(z)}{\theta_4(z)}\,,
}
 \end{array}
\eq
\beq\label{w358223}
 \begin{array}{c}
  \displaystyle{
{\bf L}'(z)_{12}
=-\breve{\nu}_0{\bf S_2 }\frac{\theta_3(\eta)}{\theta_1(\eta)} \frac{\theta_3(z)}{\theta_1(z)}+\breve{\nu}_1{\bf S_3}\frac{\theta_4(\eta)}{\theta_2(\eta)} \frac{\theta_4(z)}{\theta_2(z)}+\breve{\nu}_2{\bf S_0}\frac{\theta_1(\eta)}{\theta_3(\eta)} \frac{\theta_1(z)}{\theta_3(z)}-\breve{\nu}_3{\bf S_1}\frac{\theta_2(\eta)}{\theta_4(\eta)} \frac{\theta_2(z)}{\theta_4(z)}-
}
\\ \ \\
\displaystyle{
-\breve{\nu}_0{\bf S_3 }\frac{\theta_4(\eta)}{\theta_1(\eta)} \frac{\theta_4(z)}{\theta_1(z)}+\breve{\nu}_1{\bf S_2}\frac{\theta_3(\eta)}{\theta_2(\eta)} \frac{\theta_3(z)}{\theta_2(z)}+\breve{\nu}_2{\bf S_1}\frac{\theta_2(\eta)}{\theta_3(\eta)} \frac{\theta_2(z)}{\theta_3(z)}-\breve{\nu}_3{\bf S_0}\frac{\theta_1(\eta)}{\theta_4(\eta)} \frac{\theta_1(z)}{\theta_4(z)}\,,
}
 \end{array}
\eq
and
\beq\label{w358224}
 \begin{array}{c}
  \displaystyle{
{\bf L}'(z)_{21}
 =-\left(-\breve{\nu}_0{\bf S_2 }\frac{\theta_3(\eta)}{\theta_1(\eta)} \frac{\theta_3(z)}{\theta_1(z)}+\breve{\nu}_1{\bf S_3}\frac{\theta_4(\eta)}{\theta_2(\eta)} \frac{\theta_4(z)}{\theta_2(z)}+\breve{\nu}_2{\bf S_0}\frac{\theta_1(\eta)}{\theta_3(\eta)} \frac{\theta_1(z)}{\theta_3(z)}-\breve{\nu}_3{\bf S_1}\frac{\theta_2(\eta)}{\theta_4(\eta)} \frac{\theta_2(z)}{\theta_4(z)}\right)-
}
\\ \ \\
\displaystyle{
-\breve{\nu}_0{\bf S_3 }\frac{\theta_4(\eta)}{\theta_1(\eta)} \frac{\theta_4(z)}{\theta_1(z)}+\breve{\nu}_1{\bf S_2}\frac{\theta_3(\eta)}{\theta_2(\eta)} \frac{\theta_3(z)}{\theta_2(z)}+\breve{\nu}_2{\bf S_1}\frac{\theta_2(\eta)}{\theta_3(\eta)} \frac{\theta_2(z)}{\theta_3(z)}-\breve{\nu}_3{\bf S_0}\frac{\theta_1(\eta)}{\theta_4(\eta)} \frac{\theta_1(z)}{\theta_4(z)}\,.
}
 \end{array}
\eq
This finishes the proof. Notice that similarity between (\ref{w358221}), (\ref{w358222}) and (\ref{w358223}), (\ref{w358224}) is not accidental, the same pattern appears in the general case. The action of (\ref{w315}) or (\ref{w31513}) by conjugation on an arbitrary matrix is given in the Appendix, see (\ref{a65})-(\ref{w3582}).
\end{proof}

Obviously, the expression for the gauge transformed matrix $ {\bf {\bar L}}'(z)$ is the same as presented above
but with constants $\bar\eta$ and ${\bar\nu}_a$ instead of $\eta$ and $\nu_a$.
Then, multiplying ${\bf { L}}'(z)$ and  ${\bf {\bar L}}'(z)$ one gets an expression for $\mL'(z)$ in terms
of ${\bf S}_a={\bf S}_a(p,q,\eta)$ and ${\bf{\bar S}}_a={\bf S}_a(p,q,\eta)$ (for $\bar\eta=\eta$).
It should be mentioned that a relation of the van Diejen model
to the Sklyanin algebra was discussed at quantum level in \cite{Rains,RainsRuijs,Spiridonov}.


\section{Appendix: elliptic functions}
\def\theequation{A.\arabic{equation}}
\setcounter{equation}{0}

\paragraph{Main definitions.}
We use the standard definition of the Jacobi theta-functions
\beq\label{a01}
\begin{array}{c}
\displaystyle{
\vth(u,\tau)=\theta_1(u|\tau )=-i\sum_{k\in \mZ}
(-1)^k q^{(k+\frac{1}{2})^2}e^{\pi i (2k+1)u},
\qquad
\theta_2(u|\tau )=\sum_{k\in \mZ}
q^{(k+\frac{1}{2})^2}e^{\pi i (2k+1)u},}
\\ \\
\displaystyle{
\theta_3(u|\tau )=\sum_{k\in \mZ}
q^{k^2}e^{2\pi i ku},
\qquad
\theta_4(u|\tau )=\sum_{k\in \mZ}
(-1)^kq^{k^2}e^{2\pi i ku},}
\end{array}
\eq
where $q=e^{\pi i \tau}$ and ${\rm Im}(\tau)>0$. They are related to each other through shifts of argument
by half-periods:
\beq\label{a02}
\begin{array}{c}
\theta_2(u|\tau)=\theta_1(u+\frac{1}{2}|\tau), \! \quad \! \!
\theta_3(u|\tau)=q^{\frac{1}{4}}e^{\pi iu}\theta_1(u+\frac{\tau+1}{2}|\tau), \! \quad \!\!
\theta_4(u|\tau)=-iq^{\frac{1}{4}}e^{\pi iu}\theta_1(u+\frac{\tau}{2}|\tau).
\end{array}
\eq
The elliptic Kronecker function:
 \beq\label{a03}
  \begin{array}{l}
  \displaystyle{
 \phi(z,u)=\frac{\vth'(0)\vth(z+u)}{\vth(z)\vth(u)}=\phi(u,z)\,,\quad
 \res\limits_{z=0}\phi(z,u)=1\,,
 }
 \end{array}
 \eq
 \beq\label{a031}
  \begin{array}{l}
  \displaystyle{
\phi(-z, -u) = -\phi(z, u)
 }
 \end{array}
 \eq
is a main function.
Its derivative $f(z,u) = \partial_u \phi(z,u)$ equals
\beq\label{a04}
\begin{array}{c} \displaystyle{
    f(z, u) = \phi(z, u)(E_1(z + u) - E_1(u)), \qquad f(-z, -u) = f(z, u)\,,
}\end{array}\eq
where $E_1(z)$ the first
  Eisenstein function. The first and the second Eisenstein functions are defined as follows:
\beq\label{a05}
\begin{array}{c} \displaystyle{
    E_1(z)=\frac{\vth'(z)}{\vth(z)}=\zeta(z)+\frac{z}{3}\frac{\vth'''(0)}{\vth'(0)}\,,
    \quad
    E_2(z) = - \partial_z E_1(z) = \wp(z) - \frac{\vartheta'''(0) }{3\vartheta'(0)}\,,
}\end{array}\eq
\beq\label{a06}
\begin{array}{c}
 \displaystyle{
    E_1(- z) = -E_1(z)\,, \quad E_2(-z) = E_2(z)\,,
}\end{array}
\eq
where $\wp(z)$ and $\zeta(z)$ are the Weierstrass functions.
Also,
  \beq\label{a07}
  \begin{array}{l}
  \displaystyle{
 f(0,q)=-E_2(q)
 }
 \end{array}
 \eq
due to the local behavior of $\phi(z,q)$ near $z=0$:
  \beq\label{a08}
  \begin{array}{l}
  \displaystyle{
\phi(z,q)=z^{-1}+E_1(q)+z\,(E_1^2(q)-\wp(q))/2+O(z^2)\,.
 }
 \end{array}
 \eq
\paragraph{Periodic properties.} Define the half-periods as
 \beq\label{a09}
 \begin{array}{c}
  \displaystyle{
 \om_0=0\,,\quad
 \om_1=\frac{1}{2}\,,\quad
 \om_2=\frac{1+\tau}{2}\,,\quad
 \om_3=\frac{\tau}{2}\,.
  }
 \end{array}
\eq
 The quasi periodic properties for the defined above functions are as follows:
  \beq\label{a10}
  \begin{array}{l}
  \displaystyle{
 \phi(z+1,u)=\phi(z,u)\,,\qquad  \phi(z+\tau,u)=e^{-2\pi\imath u}\phi(z,u)
 }
 \end{array}
 \eq
or
  \beq\label{a11}
  \begin{array}{l}
  \displaystyle{
 \phi(z\pm 2\om_a,u)=e^{\mp4\pi\imath\p_\tau\om_a u}\phi(z,u)\,,
 }
 \end{array}
 \eq
which come from (for $a=1,2,3$)
  \beq\label{a12}
  \begin{array}{l}
  \displaystyle{
 \vth(z+2\om_a)=-e^{-4\pi\imath(z+\om_a)\p_\tau\om_a}\vth(z)\,,
\qquad
 \vth(z+\om_a)=-e^{-4\pi\imath z\p_\tau\om_a}\vth(z-\om_a)\,.
 }
 \end{array}
 \eq
The latter also provides
  \beq\label{a13}
  \begin{array}{l}
  \displaystyle{
 E_2(z+2\om_a)=E_2(z)\,,\qquad E_1(z+2\om_a)=E_1(z)-4\pi\imath\p_\tau\om_a
 }
 \end{array}
 \eq
and
  \beq\label{a14}
  \begin{array}{l}
  \displaystyle{
E_1(\om_a)=-2\pi\imath\p_\tau\om_a\,,\quad a\neq 0\,,
 }
 \end{array}
 \eq
  \beq\label{a141}
  \begin{array}{l}
  \displaystyle{
E_1(\om_a+\om_b)=E_1(\om_a)+E_1(\om_b)=-2\pi\imath\p_\tau(\om_a+\om_b)\,,\quad a,b\neq 0, a\neq b\,.
 }
 \end{array}
 \eq

\paragraph{Theta-functions with $2\tau$:}
%
\beq\label{a58}
\begin{array}{c}
   \displaystyle{
\theta_2(x+y|2\tau)\theta_2(x-y|2\tau)=\frac{1}{2}\Big( \theta_3(x|\tau)\theta_3(y|\tau)-
\theta_4(x|\tau)\theta_4(y|\tau) \Big)\,,
 }
 \\ \ \\
    \displaystyle{
\theta_2(x+y|2\tau)\theta_3(x-y|2\tau)=\frac{1}{2}\Big( \theta_2(x|\tau)\theta_2(y|\tau)-
\theta_1(x|\tau)\theta_1(y|\tau) \Big)\,,
 }
  \\ \ \\
   \displaystyle{
\theta_3(x+y|2\tau)\theta_3(x-y|2\tau)=\frac{1}{2}\Big( \theta_3(x|\tau)\theta_3(y|\tau)+
\theta_4(x|\tau)\theta_4(y|\tau) \Big)\,.
 }
\end{array}
\eq
In particular, for the matrix $\Xi(z)$ (\ref{w315}) we have
\beq\label{a59}
\begin{array}{c}
   \displaystyle{
\det\Xi(z)=-\theta_1(z)\theta_1(2q)\,.
 }
\end{array}
\eq

\paragraph{Riemann identities:}
\beq\label{a60}
\begin{array}{c}
\displaystyle{
\theta_\al(z-2q)\theta_1(2q+z)\theta_\be(0)\theta_\ga(0)=\theta_1(2q)\theta_\al(2q)\theta_\be(z)\theta_\ga(z)
+\theta_1(z)\theta_\al(z)\theta_\be(2q)\theta_\ga(2q)\,,
}
\end{array}
\eq
where $\al, \be, \ga \in\{2,3,4\}$ and distinct.
 The rest of identities can be found in \cite{Mum}. Inter alia, it follows from (\ref{a60}) that
\beq\label{a61}
\displaystyle{
\theta_\al(z-2q)\phi(2q,z)-\theta_\al(z+2q)\phi(-2q,z) = \frac{2\theta_1'(0)\theta_\al(z)\theta_\be(2q)\theta_\ga(2q)}{\theta_\be(0)\theta_\ga(0)\theta_1(2q)}\,.
}
\eq
We also use particular cases of the above identities known as the Weierstrass identities (see e.g. \cite{KhZ}):
\begin{equation}\label{w623}
\begin{aligned}
&\theta_{1}(u+x)\theta_{1}(u-x)\theta_{r}(v+y)\theta_{r}(v-y) - \theta_{1}(v+x)\theta_{1}(v-x)\theta_{r}(u+y)\theta_{r}(u-y) \\
&= \theta_{1}(u+v)\theta_{1}(u-v)\theta_{r}(x+y)\theta_{r}(x-y), \quad r=1,2,3,4.
\end{aligned}
\end{equation}
\begin{equation}\label{w624}
\begin{aligned}
&\theta_{2}(u+x)\theta_{2}(u-x)\theta_{3}(v+y)\theta_{3}(v-y) - \theta_{2}(v+x)\theta_{2}(v-x)\theta_{3}(u+y)\theta_{3}(u-y) \\
&= -\theta_{1}(u+v)\theta_{1}(u-v)\theta_{4}(x+y)\theta_{4}(x-y),
\end{aligned}
\end{equation}
\begin{equation}\label{w625}
\begin{aligned}
&\theta_{2}(u+x)\theta_{2}(u-x)\theta_{4}(v+y)\theta_{4}(v-y) - \theta_{2}(v+x)\theta_{2}(v-x)\theta_{4}(u+y)\theta_{4}(u-y) \\
&= -\theta_{1}(u+v)\theta_{1}(u-v)\theta_{3}(x+y)\theta_{3}(x-y),
\end{aligned}
\end{equation}
\begin{equation}\label{w626}
\begin{aligned}
&\theta_{3}(u+x)\theta_{3}(u-x)\theta_{4}(v+y)\theta_{4}(v-y) - \theta_{3}(v+x)\theta_{3}(v-x)\theta_{4}(u+y)\theta_{4}(u-y) \\
&= -\theta_{1}(u+v)\theta_{1}(u-v)\theta_{2}(x+y)\theta_{2}(x-y).
\end{aligned}
\end{equation}
\begin{equation}\label{w627}
\begin{aligned}
&\theta_{r}(u+x)\theta_{r}(u-x)\theta_{r}(v+y)\theta_{r}(v-y) - \theta_{r}(u+y)\theta_{r}(u-y)\theta_{r}(v+x)\theta_{r}(v-x) \\
&= (-1)^{r-1} \theta_{1}(u+v)\theta_{1}(u-v)\theta_{1}(x+y)\theta_{1}(x-y), \quad r=1,2,3,4.
\end{aligned}
\end{equation}

\paragraph{Addition formulae for the function $\phi$.} Main summation formula is
\beq\label{a15}
  \begin{array}{c}
  \displaystyle{
  \phi(z_1, u_1) \phi(z_2, u_2) = \phi(z_1, u_1 + u_2) \phi(z_2 - z_1, u_2) + \phi(z_2, u_1 + u_2) \phi(z_1 - z_2, u_1)\,.
 }
 \end{array}
 \eq
It leads to
\beq\label{a16}
  \begin{array}{c}
  \displaystyle{
  \phi(z_1, u_1) f(z_2, u_2)-f(z_1,u_1)\phi(z_2,u_2) =
  }
  \\ \ \\
  \displaystyle{
  =
  \phi(z_1, u_1 + u_2) f(z_2 - z_1, u_2) - \phi(z_2, u_1 + u_2) f(z_1 - z_2, u_1)\,.
 }
 \end{array}
 \eq
Also, in certain limiting cases (\ref{a15}) yields
\beq\label{a17}
  \begin{array}{c}
  \displaystyle{
 \phi(z,u_1)\phi(z,u_2)=\phi(z,u_1+u_2)\Big(E_1(z)+E_1(u_1)+E_1(u_2)-E_1(z+u_1+u_2)\Big)\,,
 }
 \end{array}
 \eq
\beq\label{a18}
  \begin{array}{c}
  \displaystyle{
  \phi(z, u_1) f(z,u_2)-\phi(z, u_2) f(z,u_1)=\phi(z,u_1+u_2)\Big(\wp(u_1)-\wp(u_2)\Big)=
 }
 \\ \ \\
   \displaystyle{
 =\phi(z,u_1+u_2)\Big(E_2(u_1)-E_2(u_2)\Big)=\phi(z,u_1+u_2)\Big(f(0,u_2)-f(0,u_1)\Big)\,.
 }
 \end{array}
 \eq
\beq\label{a19}
  \begin{array}{c}
  \displaystyle{
  \phi(z, u) \phi(z, -u) = \wp(z)-\wp(u)=E_2(z)-E_2(u)\,,
 }
 \end{array}
 \eq
\beq\label{a20}
  \begin{array}{c}
  \displaystyle{
  \phi(z, u) f(z, -u)-\phi(z, -u) f(z, u)=\wp'(u)\,.
 }
 \end{array}
 \eq

\paragraph{The functions $\vf_k(z,x+\om_k)$.} Using the function $\phi$ and the notation (\ref{a03}) define
the following four functions:
 \beq\label{a21}
 \begin{array}{c}
  \displaystyle{
 \vf_{0}(z,x)=\phi(z,x)\,,
\qquad
 \vf_{1}(z,x+\om_1)=\phi(z,x+\om_1)\,,
 }
 \\ \ \\
  \displaystyle{
 \vf_{2}(z,x+\om_2)=e^{\pi\imath z}\phi(z,x+\om_2)\,,
\qquad
  \vf_{3}(z,x+\om_3)=e^{\pi\imath z}\phi(z,x+\om_3)\,,
 }
  \end{array}
 \eq
or equivalently,
 \beq\label{a22}
 \begin{array}{c}
  \displaystyle{
 \vf_{k}(z,x+\om_k)=e^{2\pi\imath z\p_\tau\om_k}\phi(z,x+\om_k)\,,
 \quad k=0,1,2,3
  }
 \end{array}
\eq
 or
 \beq\label{a222}
 \begin{array}{c}
  \displaystyle{
 \vf_{k}(z,x+\om_k)=\frac{\theta_1'(0)\theta_{k+1}(z+x)}{\theta_{1}(z)\theta_{k+1}(x)}\,,
 \quad k=0,1,2,3\,.
  }
 \end{array}
\eq
Due to (\ref{a031}) and (\ref{a22})
 \beq\label{a223}
 \begin{array}{c}
  \displaystyle{
 \vf_{k}(-z,x+\om_k)=-\vf_{k}(z,-x+\om_k)\,,
 \quad k=0,1,2,3\,.
  }
 \end{array}
\eq
 The action of the matrix $I$ from (\ref{w2113}):
  \beq\label{a23}
 \begin{array}{c}
  \left(\begin{array}{c}
 \vf_{0}(2z,u+\om_0)
 \\
  \vf_{1}(2z,u+\om_1)
 \\
  \vf_{2}(2z,u+\om_2)
 \\
  \vf_{3}(2z,u+\om_3)
 \end{array}\right)
 =
   \displaystyle{\frac12}
 \left(\begin{array}{cccc}
 1 & 1 & 1 & 1
 \\
  1 & 1 & -1 & -1
 \\
  1 & -1 & 1 & -1
 \\
  1 & -1 & -1 & 1
 \end{array}\right)
  \left(\begin{array}{c}
 \vf_{0}(2u,z+\om_0)
 \\
  \vf_{1}(2u,z+\om_1)
 \\
  \vf_{2}(2u,z+\om_2)
 \\
  \vf_{3}(2u,z+\om_3)
 \end{array}\right)\,.
 \end{array}
 \eq
  This transformation matrix is written as
  \beq\label{a24}
  \begin{array}{c}
  \displaystyle{
I_{km}=\frac12\exp\Big(4\pi\imath\Big(\om_{m-1}\p_\tau\om_{k-1}
-\om_{k-1}\p_\tau\om_{m-1}\Big)\Big)\,,\quad k,m=1,...,4
 }
 \end{array}
 \eq
and satisfies the property
  \beq\label{a25}
  \begin{array}{c}
  \displaystyle{
I^{-1}=I\,.
 }
 \end{array}
 \eq
  We also use the widely known identities:
  \beq\label{a26}
 \begin{array}{c}
  \displaystyle{
\sum\limits_{a=0}^3\wp(z+\om_a)=4\wp(2z)
  }
 \end{array}
 \eq
 and for $\al=1,2,3$
  \beq\label{a261}
 \begin{array}{c}
  \displaystyle{
\wp(z+\om_\al)-\wp(\om_\al)=\frac{1}{2}\,\frac{\wp''(\om_\al)}{\wp(z)-\wp(\om_\al)}\,.
  }
 \end{array}
 \eq
\paragraph{The functions $\vf_k(z)$ and theta-constants.}
Here we discuss the properties of the functions (\ref{a222}) at $x=0$ for $k=1,2,3$:
$\vf_{k}(z)=\vf_{k}(z,\om_k)=e^{2\pi\imath z\p_\tau\om_k}\phi(z,\om_k)$ or
 \beq\label{a27}
 \begin{array}{c}
  \displaystyle{
 \vf_{1}(z)=\frac{\theta_1'(0)\theta_{2}(z)}{\theta_{1}(z)\theta_{2}(0)}\,,\qquad
  \vf_{2}(z)=\frac{\theta_1'(0)\theta_{3}(z)}{\theta_{1}(z)\theta_{3}(0)}\,,\qquad
   \vf_{3}(z)=\frac{\theta_1'(0)\theta_{4}(z)}{\theta_{1}(z)\theta_{4}(0)}\,.
  }
 \end{array}
\eq
Due to (\ref{a223}) these functions are odd:
 \beq\label{a28}
 \begin{array}{c}
  \displaystyle{
 \vf_{k}(-z)=-\vf_{k}(z)\,,
 \quad k=1,2,3\,.
  }
 \end{array}
\eq
 According to numeration of half-periods (\ref{a09}) introduce notation $\wp_k=\wp(\om_k)$,
 $k=1,2,3$:
 \beq\label{a29}
 \begin{array}{c}
  \displaystyle{
 \wp_1=\wp\Big(\frac{1}{2}\Big)\,,\qquad
 \wp_2=\wp\Big(\frac{1+\tau}{2}\Big)\,,\qquad
 \wp_3=\wp\Big(\frac{\tau}{2}\Big)\,.
  }
 \end{array}
\eq
As is widely known
 \beq\label{a291}
 \begin{array}{c}
  \displaystyle{
 \wp_1+\wp_2+\wp_3=0\,.
  }
 \end{array}
\eq
From (\ref{a19}) one gets
 \beq\label{a30}
 \begin{array}{c}
  \displaystyle{
 \vf_{k}^2(z)=\wp(z)-\wp_k\,, \quad k=1,2,3\,,
  }
 \end{array}
\eq
 and, therefore,
 \beq\label{a31}
 \begin{array}{c}
  \displaystyle{
 \vf_{\al}^2(z)-\vf_{\be}^2(z)=\wp_\be-\wp_\al\,,\qquad \al,\be\in\{1,2,3\}\,.
  }
 \end{array}
\eq
 Let $\al,\be,\ga\in\{1,2,3\}$ be pairwise distinct. For the derivative we have
 \beq\label{a32}
 \begin{array}{c}
  \displaystyle{
 \p_z\vf_{\al}(z)=\vf'_{\al}(z)=\vf_\al(E_1(z+\om_\al)-E_1(z)-E_1(\om_\al))=-\vf_\be(z)\vf_\ga(z)\,.
  }
 \end{array}
\eq
 Then, by differentiating (\ref{a30}) one obtains
 \beq\label{a33}
 \begin{array}{c}
  \displaystyle{
 -2\vf_1(z)\vf_2(z)\vf_3(z)=\wp'(z)\,.
  }
 \end{array}
\eq
 Its square yields the equation of the elliptic curve:
 \beq\label{a34}
 \begin{array}{c}
  \displaystyle{
 (\wp'(z))^2=4(\wp(z)-\wp_1)(\wp(z)-\wp_2)(\wp(z)-\wp_3)\,.
  }
 \end{array}
\eq
 Introduce the following theta-constants $c_k=c_k(\tau)$, $k=1,2,3$:
 \beq\label{a35}
 \begin{array}{c}
  \displaystyle{
 c_1(\tau)=-\Big(\frac{\vth_2(0)}{\vth'(0)}\Big)^2\,,\quad
 c_2(\tau)=-\imath\Big(\frac{\vth_3(0)}{\vth'(0)}\Big)^2\,,\quad
 c_3(\tau)=-\Big(\frac{\vth_4(0)}{\vth'(0)}\Big)^2\,.
  }
 \end{array}
\eq
 Equivalently,
 \beq\label{a36}
 \begin{array}{c}
  \displaystyle{
 c_1(\tau)=-\frac{1}{(\wp_1-\wp_2)^{1/2}(\wp_1-\wp_3)^{1/2}}
 =-\frac{1}{\vf_2(\om_1)\vf_3(\om_1)}=-\Big(\frac{\vth(\om_1)}{\vth'(0)}\Big)^2\,,
  }
  \\ \ \\
    \displaystyle{
 c_2(\tau)=\frac{1}{(\wp_2-\wp_1)^{1/2}(\wp_2-\wp_3)^{1/2}}=
 -\frac{1}{\vf_1(\om_2)\vf_3(\om_2)}=-e^{\pi\imath\om_2}\Big(\frac{\vth(\om_2)}{\vth'(0)}\Big)^2\,,
  }
    \\ \ \\
    \displaystyle{
 c_3(\tau)=\frac{1}{(\wp_3-\wp_1)^{1/2}(\wp_3-\wp_2)^{1/2}}=
  \frac{1}{\vf_1(\om_3)\vf_2(\om_3)}=e^{\pi\imath\om_3}\Big(\frac{\vth(\om_3)}{\vth'(0)}\Big)^2
  }
 \end{array}
\eq
 and
 \beq\label{a361}
 \begin{array}{c}
  \displaystyle{
 \vf_1(z+\om_1)=\frac{1}{c_1\vf_1(z)}\,,\qquad
  \vf_2(z+\om_2)=\frac{1}{c_2\vf_2(z)}\,,\qquad
    \vf_3(z+\om_3)=-\frac{1}{c_3\vf_3(z)}\,.
  }
 \end{array}
\eq
 These set of constants satisfy identities
 \beq\label{a37}
 \begin{array}{c}
  \displaystyle{
\sum\limits_{k=1}^3c_k^2=0\,,
\qquad
\sum\limits_{k=1}^3c_k^2\wp_k=0\,,
  }
  \\
  \displaystyle{
\sum\limits_{k=1}^3c_k^2\wp_k^2=1\,,
\qquad
c_1^2\wp_2\wp_3+c_2^2\wp_1\wp_3+c_3^2\wp_1\wp_2=1
  }
 \end{array}
\eq
 and
 \beq\label{a38}
 \begin{array}{c}
  \displaystyle{
c_\al c_\be=-\imath\varepsilon_{\al\be\ga}\frac{c_\ga}{\wp_\al-\wp_\be}\,.
  }
 \end{array}
\eq

\paragraph{Properties of the function $v(z,u)$.}
Here we list some identities for the function (\ref{w210}). More properties
can be found in  \cite{KH,IK,Hikami,Z04,Ch2} and in the Appendix from \cite{MMZ}.

First,
  \beq\label{a50}
  \begin{array}{c}
  \displaystyle{
v(z,u)v(z,-u)=\sum\limits_{a=0}^3\Big(\breve\nu_a^2\wp(z+\om_a)-\nu_a^2\wp(u+\om_a)\Big)\,.
 }
 \end{array}
 \eq
Also,
    \beq\label{vbarv2}
    \begin{array}{c}
    \displaystyle{
    v(z,q)\bar v(z,-q) + v(z,-q)\bar v(z,q)=
   2\sum_{\alpha=0}^3\breve\nu_\alpha\breve{\bar\nu}_\alpha\wp(z+\om_\al)-
    2\sum_{\alpha=0}^3\nu_\alpha{\bar\nu}_\alpha\wp(q+\om_\al)\,.
    }
    \end{array}
    \eq
The latter is obtained by writing (\ref{a50}) for $\nu\to \nu+\bar\nu$ and then subtracting 
(\ref{a50}) for $\nu$ and $\bar\nu$ from the obtained result.

Due to (\ref{a223}) we have
\beq\label{a51}
v(-z,-u)=-v(z,u)
\eq
and, therefore,
\beq\label{a52}
v(z,u)v(-z,u)=\sum\limits_{a=0}^3\Big(\nu_a^2\wp(u+\om_a)-\breve\nu_a^2\wp(z+\om_a)\Big)=-v(z,u)v(z,-u)\,.
\eq
Using (\ref{a04}) introduce
  \beq\label{a53}
  \begin{array}{c}
  \displaystyle{
v'(z,u)=\p_uv(z,u)=\sum\limits_{a=0}^3\nu_a \exp(4\pi \imath z\p_\tau\om_a)f(2z,u+\om_a)\,.
 }
 \end{array}
 \eq
Then from (\ref{a50}) we conclude that
  \beq\label{a54}
  \begin{array}{c}
  \displaystyle{
v(z,u)v'(z,-u)-v(z,-u)v'(z,u)=\sum\limits_{a=0}^3\nu_a^2\wp'(u+\om_a)\,,
 }
 \end{array}
 \eq
and from (\ref{a07}) one gets
  \beq\label{a55}
  \begin{array}{c}
  \displaystyle{
v'(0,u)=\sum\limits_{a=0}^3\nu_a f(0,u+\om_a)=-\sum\limits_{a=0}^3\nu_a E_2(u+\om_a)=
-\sum\limits_{a=0}^3\nu_a \wp(u+\om_a)+\frac{\vth'''(0)}{3\vth'(0)}\sum\limits_{a=0}^3\nu_a\,.
 }
 \end{array}
 \eq
Also,
\beq\label{a551}
  \begin{array}{c}
  \displaystyle{
\sum\limits_{k=0}^3\nu_k^2\wp(q+\om_k)=(\tnu_0^2+\tnu_1^2+\tnu_2^2+\tnu_3^2)\wp(2q)+
}
\\ \ \\
  \displaystyle{
+2\tnu_0\Big(\tnu_1\vf_2(2q)\vf_3(2q)+\tnu_2\vf_1(2q)\vf_3(2q)+\tnu_3\vf_1(2q)\vf_2(2q)\Big)+
}
\\ \ \\
  \displaystyle{
  +2\Big(\tnu_1\tnu_2\vf_1(2q)\vf_2(2q)+\tnu_2\tnu_3\vf_2(2q)\vf_3(2q)+\tnu_1\tnu_3\vf_1(2q)\vf_3(2q)\Big)\,.
 }
 \end{array}
 \eq
The latter follows from comparing  (\ref{w370}) and (\ref{a50}).

Summation formulae are as follows:
\beq\label{a56}
  \begin{array}{c}
  \displaystyle{
v(x,u)\phi(x+y,w-u)+v(x,w)\phi(x-y,u-w)+v(y,-u)\phi(x+y,u+w)=
 }
 \\ \ \\
   \displaystyle{
=v(y,w)\phi(x-y,u+w)
}
\end{array}
\eq
and
\beq\label{a57}
\begin{array}{c}
   \displaystyle{
\phi(z,u-w)(v'(0,w)-v'(0,u))=
 }
\end{array}
\eq
$$
   \displaystyle{
2v(-z,w)f(z,u+w)+2v(z,u)f(-z,u+w)+v'(-z,w)\phi(z,u+w)+v'(z,u)\phi(-z,u+w)\,.
 }
$$

We also use the identity
\beq\label{a62}
\begin{array}{c}
\displaystyle{
\phi(z - w, v)\phi(z, u_1 - v)\phi(w, u_2 + v) - \phi(z - w, u_1 - u_2 - v)\phi(z, u_2 + v)\phi(w, u_1 - v)=}
\\ \ \\
\displaystyle{
= \phi(z, u_1)\phi(w, u_2)\Big(E_1(v)-E_1(u_1- u_2- v) + E_1(u_1- v)-E_1(u_2 + v)\Big)}\,.
\end{array}
\eq
In the limiting case $z=w$ it follows from (\ref{a62}) (using also (\ref{a08})) that
\beq\label{a621}
\begin{array}{c}
\displaystyle{
- \phi(z, u_1 - v)\phi'(z, u_2 + v) + \phi(z, u_2 + v)\phi'(z, u_1 - v) +
}
\\ \ \\
\displaystyle{
+\phi(z, u_2 + v)\phi(z, u_1 - v)\Big(E_1(v)-E_1(u_1- u_2- v)\Big)=
}
\\ \ \\
\displaystyle{
= \phi(z, u_1)\phi(z, u_2)\Big(E_1(v)-E_1(u_1- u_2- v) + E_1(u_1- v)-E_1(u_2 + v)\Big)}\,,
\end{array}
\eq
where $\phi'(z,u)=\p_z\phi(z,u)$. Using (\ref{a13})-(\ref{a141}) one gets
\beq\label{a63}
\begin{array}{c}
\displaystyle{
\vf_k(2q,\eta+\om_k)\vf'_j(2q,\eta+\om_j)-\vf'_k(2q,\eta+\om_k)\vf_j(2q,\eta+\om_j)=
}
\\ \ \\
\displaystyle{
=\phi(2q,\eta)\vf_i(2q,\eta+\om_i)\Big( E_1(\eta+\om_j)- E_1(\eta+\om_k)+ E_1(\om_k)- E_1(\om_j)\Big)\,,
}
\end{array}
\eq
where prime is a derivative with respect to the first argument, and $i,j,k\in\{1,2,3\}$ are distinct.
And
\beq\label{a64}
\begin{array}{c}
\displaystyle{
\phi(2q,\eta)\vf'_k(2q,\eta+\om_k)-\phi'(2q,\eta)\vf_k(2q,\eta+\om_k)=
}
\\ \ \\
\displaystyle{
=\vf_j(2q,\eta+\om_j)\vf_i(2q,\eta+\om_i)\Big( E_1(\eta+\om_k)- E_1(\eta)- E_1(\om_k)\Big)\,.
}
\end{array}
\eq
\paragraph{Gauge transformation.}
Consider an arbitrary matrix of the form
\beq\label{a65}
A\equiv A(q,p,z) = \mat{a(q,p,z)}{b(q,p,z)}{b(-q,-p,z)}{a(-q,-p,z)}\,.
\eq
Conjugation with the matrix $\Xi(z)$ (\ref{w315})\footnote{The matrix (\ref{w31513}) differs by only  shift of variable $z$ to $z+\eta$.} acts as follows:
\beq\label{w3581}
 \begin{array}{c}
  \displaystyle{
(\Xi(z)A\Xi^{-1}(z))_{11}=\frac{1}{\det\Xi}\Big[\theta_2(z+2q|2\tau)\Big(\theta_3(z-2q|2\tau)a(q,p,z)-
}
\\ \ \\
\displaystyle{
-\theta_3(z+2q|2\tau)b(-q,-p,z)\Big)+\theta_2(z-2q|2\tau)\Big(\theta_3(z-2q|2\tau)b(q,p,z)-
}
\\ \ \\
\displaystyle{
-\theta_3(z+2q|2\tau)a(-q,-p,z)\Big)\Big]=\frac{-1}{2\theta_1(z)\theta_1(2q)}\Big[a(q,p,z)\theta_2(z)\theta_2(2q)+
 }
 \\ \ \\
 \displaystyle{
 +b(q,p,z)\theta_2(0)\theta_2(z-2q)\Big]+\frac12a(q,p,z)+(q\rightarrow{-q}, p\rightarrow{-p})
 }
 \end{array}
 \eq
 Here we used (\ref{a58}), the oddness of $\theta_1(z)$ and the evenness of $\theta_a(z)$ for $a \in\{2,3,4\}$.
 For other entries we have:
 \beq\label{w3582}
 \begin{array}{c}
  \displaystyle{
(\Xi(z)A\Xi^{-1}(z))_{22}=\frac{1}{2\theta_1(z)\theta_1(2q)}\Big[a(q,p,z)\theta_2(z)\theta_2(2q)+
 }
 \\ \ \\
 \displaystyle{
 +b(q,p,z)\theta_2(0)\theta_2(z-2q)\Big]+\frac12a(q,p,z)+(q\rightarrow{-q}, p\rightarrow{-p})\,,
 }
 \\ \ \\
  \displaystyle{
(\Xi(z)A\Xi^{-1}(z))_{12}=\frac{1}{2\theta_1(z)\theta_1(2q)}\Big[a(q,p,z)\Big(\theta_3(z)\theta_3(2q)+\theta_4(z)\theta_4(2q)\Big)
 }
 \\ \ \\
 \displaystyle{
 +b(q,p,z)\Big(\theta_3(0)\theta_3(z-2q)+\theta_4(0)\theta_4(z-2q)\Big)\Big]+(q\rightarrow{-q}, p\rightarrow{-p})\,,
 }
 \\ \ \\
  \displaystyle{
(\Xi(z)A\Xi^{-1}(z))_{21}=\frac{1}{2\theta_1(z)\theta_1(2q)}\Big[a(q,p,z)\Big(-\theta_3(z)\theta_3(2q)+\theta_4(z)\theta_4(2q)\Big)
 }
 \\ \ \\
 \displaystyle{
 +b(q,p,z)\Big(-\theta_3(0)\theta_3(z-2q)+\theta_4(0)\theta_4(z-2q)\Big)\Big]+(q\rightarrow{-q}, p\rightarrow{-p})\,.
 }
 \end{array}
 \eq


\subsection*{Acknowledgments}
This work was supported by the Russian Science Foundation under grant no. 25-11-00081,\\
https://rscf.ru/en/project/25-11-00081/ and performed at Steklov Mathematical Institute of Russian Academy of Sciences.






\begin{small}

\end{small}

\end{document}